\documentclass[10pt]{article}
\usepackage[letterpaper,margin=1in]{geometry}
\usepackage{graphicx} %

\usepackage{xcolor} %
\usepackage{ninecolors}\NineColors{saturation=high}
\usepackage{amsmath, amssymb, amsfonts, braket,  enumerate, mathrsfs, tikz, caption, subcaption, xcolor, mathtools, braket, dirtytalk}
\usepackage{bbm, dsfont}
\usepackage{amsthm}
\usepackage{thmtools} 
\usepackage{thm-restate}

\usepackage{tabularx}

\usepackage[linesnumbered,ruled,vlined]{algorithm2e}
\SetKwInput{KwInput}{Input}                %
\SetKwInput{KwOutput}{Output}
\SetKwFor{RepTimes}{repeat}{times}{end}

\usepackage[pagebackref,hypertexnames=false]{hyperref}
\hypersetup{
    colorlinks=true,
    urlcolor=blue3,
    linkcolor=green3,
    citecolor=red3,
}

\usepackage[nameinlink,capitalize]{cleveref}

\renewcommand{\backref}[1]{}

\renewcommand{\backrefalt}[4]{%
\ifcase #1 %
\or
[p.\ #2]%
\else
[pp.\ #2]%
\fi}

\newcommand{\SWAP}{\mathrm{SWAP}}

\newcommand{\nn}{\nonumber\\}

\DeclareMathOperator*{\E}{\mathbf{E}}
\newcommand{\PREP}{\mathrm{PREP}}
\newcommand{\SEL}{\mathrm{SEL}}

\newcommand{\wh}[1]{\widehat{#1}}

\newcommand{\F}{\mathbb{F}}
\renewcommand{\hat}{\widehat}

\theoremstyle{plain}
\newtheorem{theorem}{Theorem}[section]
\newtheorem{corollary}[theorem]{Corollary}
\newtheorem{definition}[theorem]{Definition}
\newtheorem{proposition}[theorem]{Proposition}
\newtheorem{remark}[theorem]{Remark}
\newtheorem{lemma}[theorem]{Lemma}

\crefname{claim}{Claim}{Claims}
\newtheorem{fact}[theorem]{Fact}
\crefname{fact}{Fact}{Facts}

\newcommand{\Pmod}{\overline{\mathcal P}_n}
\newcommand{\N}{\mathbb{N}}
\newcommand{\diag}{\mathrm{diag}}

\newcommand{\eps}{\varepsilon}

\newcommand{\dist}{\mathsf{dist}}
\newcommand{\distphop}{\dist_{\rm{phaseop}}}
\newcommand{\distdiamond}{\dist_{\diamond}}

\renewcommand{\Pr}{\mathop{\bf Pr\/}}
\newcommand{\Ex}{\mathop{\bf E\/}}

\newcommand{\tr}{\mathrm{tr}} \newcommand{\Tr}{\tr}

\newcommand{\supp}{\mathrm{supp}}

\newcommand{\calA}{\mathcal{A}}

\newcommand{\calC}{\mathcal{C}}
\newcommand{\calD}{\mathcal{D}}

\newcommand{\calP}{\mathcal{P}}

\newcommand{\calS}{\mathcal{S}}

\newcommand{\calW}{\mathcal{W}}

\newcommand{\C}{\mathbb C}

\newcommand{\poly}{\mathsf{poly}}

\newcommand{\abs}[1]{\lvert #1 \rvert}

\newcommand{\norm}[1]{\lVert #1 \rVert}
\newcommand{\opnorm}[1]{\norm{#1}_{\mathrm{op}}}

\newcommand{\fnorm}[1]{\norm{#1}_{F}}

\newcommand{\ketbra}[2]{\ket{#1}\!\!\bra{#2}}

\newcommand{\sympcomp}{\perp}
\newcommand{\paulisupport}{\calW}

\title{Efficient Learning of Structured Quantum Circuits\\via Pauli Dimensionality and Sparsity}
\author{}
\author{Sabee Grewal\thanks{\texttt{sabee@cs.utexas.edu}. The University of Texas at Austin.} 
\and Daniel Liang\thanks{\texttt{daniel.liang@ll.mit.edu}. Portland State University.}}
\date{}

\begin{document}

\maketitle
\begin{abstract}
We study the problem of efficiently learning an unknown $n$-qubit unitary channel in \emph{diamond distance} given query access.
We present a general framework showing that if Pauli operators remain low-complexity under conjugation by a unitary, then the unitary can be learned efficiently.
This framework yields polynomial-time algorithms for a wide range of circuit classes, including $O(\log \log n)$-depth circuits, quantum $O(\log n)$-juntas, near-Clifford circuits, the Clifford hierarchy, fermionic matchgate circuits, and certain compositions thereof. 
Our results unify and generalize prior work, and yield efficient learning algorithms for more expressive circuit classes than were previously known. 

Our framework is powered by new learning algorithms for unitaries whose Pauli spectrum is either supported on a small subgroup or is sparse. 
If the Pauli spectrum is supported on a subgroup of size $2^k$, we give an $\widetilde{O}(2^k/\eps)$-query algorithm and a nearly matching $\Omega(2^k/\eps)$ lower bound. 
For $k = 2n$, we recover the optimal $O(4^n/\eps)$-query algorithm of Haah, Kothari, O’Donnell, and Tang [FOCS ’23].
If the Pauli spectrum is supported on $s$ Pauli operators, we give an $O(s^2/\eps^2)$-query algorithm and an $\Omega(s/\eps)$ lower bound.
\end{abstract}

\newpage 
\hypersetup{linktocpage}
\tableofcontents

\newpage

\section{Introduction}
Given black-box access to an unknown unitary process, how efficiently can one reconstruct its action? This task---known as \emph{unitary process tomography}---is a central problem in quantum information and quantum computation.
It has been extensively studied over the past several decades under various models and performance measures (see \cite[Section 1.3]{haah2023query} for a detailed overview).
Recently, Haah, Kothari, O’Donnell, and Tang~\cite{haah2023query} established that $\Theta(4^n/\eps)$ queries are both necessary and sufficient to learn an unknown $n$-qubit unitary to $\eps$ accuracy in diamond norm. 

A complementary challenge is to identify structured subclasses of unitary channels that admit \emph{efficient} learning algorithms in both query complexity and runtime. 
This direction seeks to characterize which quantum dynamics are tractable to learn, and is practically relevant for the experimental verification of quantum devices. 
Moreover, efficiently learnable circuit classes cannot be pseudorandom, so such algorithms delineate the boundary of pseudorandomness for quantum circuits, a direction that has received significant recent attention~\cite{ma2024note,ma2025howto,schuster2025random,lu2025parallelkacswalkgenerates,pmlr-v291-chia25a,foxman2025randomunitariesconstantquantum}.
For these reasons, a growing line of work has identified efficiently learnable subclasses of quantum circuits, including constant-depth circuits~\cite{huang2024shallow}, fermionic matchgates~\cite{oszmaniec2022fermion}, Clifford circuits~\cite{low2009learning}, and other restricted models. 

Despite this progress, existing efficient-learning results are tailored to specific circuit classes and rely on disparate techniques. 
This raises a basic question: is there a common structural reason why many circuit classes are learnable?
In this work, we answer this question by developing a general framework showing that if a generating set of Pauli operators remains low-complexity under conjugation by a unitary, then the unitary can be learned efficiently.
This perspective provides a unifying explanation for a broad range of prior results, extends them to more expressive circuit classes, and in several cases yields algorithms with improved query and time complexity. 

\subsection{The Unifying Framework}

Our framework studies the action of the unknown unitary $U$ on Pauli operators via conjugation, i.e., the map $P \mapsto U^\dagger P U$. 
Since the Pauli operators form a basis for all $n$-qubit operators, this action fully determines $U$. 
In particular, it suffices to understand the action of $U$ on a \emph{generating set} of Pauli operators. 
A set $G$ of Pauli operators is said to generate the Pauli group if every Pauli operator can be written as a product of elements of $G$ (up to global phase). 
Our framework shows that if the operators $\{U^\dagger g U : g \in G\}$ have low-complexity Pauli spectra, then $U$ can be learned efficiently.

We consider two natural notions of complexity of the Pauli spectrum. Any unitary $U$ admits a Pauli expansion $U = \sum_{P \in \{I,X,Y,Z\}^{\otimes n}} \alpha_P P$, where the coefficients $\{\alpha_P\}$ form the Pauli spectrum of $U$. We say that $U$ has \emph{Pauli dimensionality} $k$ if its spectrum is supported on a subgroup $G$ of size $2^k$, and \emph{Pauli sparsity} $s$ if its spectrum is supported on a set $S$ of size $s$. 
These notions can be viewed as quantum analogues of Fourier dimensionality and sparsity in Boolean function analysis~\cite{gopalan2011testing}.

Let $\distdiamond(\cdot,\cdot)$ denote the diamond distance, i.e., the standard worst-case distance measure between quantum channels. 

\begin{theorem}[Combination of \cref{thm:generating-set-suffices,thm:efficient-learning-thm}]
\label{thm:intro-framework}
Let $G$ be a known generating set of Pauli operators, and let $U$ be an $n$-qubit unitary such that for every $g \in G$, the operator $U^\dagger g U$ has Pauli dimensionality $k$ (resp.\ Pauli sparsity $s$). Then there is an algorithm that outputs a unitary $V$ satisfying $\distdiamond(U, V) \leq \eps$ with probability at least $1-\delta$. The algorithm uses $\poly(2^k, n)\cdot \frac{\log(1/\delta)}{\eps}$ (resp.\ $\poly(s, n)\cdot \frac{\log(1/\delta)}{\eps^2}$) queries.
\end{theorem}

Thus, \cref{thm:intro-framework} reduces unitary learning to identifying a generating set $G$ for which $\{U^\dagger g U : g \in G\}$ have low-complexity Pauli spectra. 
We therefore obtain learning algorithms for broad classes of quantum circuits by establishing this property.

Among the core technical ingredients underlying \cref{thm:intro-framework} are new learning algorithms for learning unitaries with low-complexity Pauli spectra. 

\begin{theorem}[Informal version of \cref{cor:pauli-dimension-bootstrap,thm:learning-sparsity,cor:dimension-lowerbound,cor:pauli-sparsity-LB}]
\label{thm:intro-pauli-dimension-sparsity}
    Let $U$ be an $n$-qubit $k$-Pauli-dimensional (resp.\ $s$-Pauli-sparse) unitary. 
    There is an efficient algorithm that outputs a unitary $V$ satisfying $\distdiamond(U,V) \leq \eps$ with high probability using $O(2^k k/\eps)$  (resp. $O(s^2/\eps^2)$) queries. 
Moreover, any algorithm requires $\Omega(2^k/\eps)$ (resp. $\Omega(s/\eps)$) queries.
\end{theorem}

\cref{thm:intro-pauli-dimension-sparsity} provides the core primitives underlying \cref{thm:intro-framework}. 
In particular, the efficiency of the resulting learning algorithms is governed by the complexity of these subroutines, and improvements to these primitives directly translate into improved algorithms for all circuit classes captured by our framework.
This gives strong motivation for optimizing the performance of these subroutines.

We make several remarks about \cref{thm:intro-framework,thm:intro-pauli-dimension-sparsity}.
For unitaries with Pauli dimensionality $k$, the resulting algorithms are \emph{time efficient} in all cases.
For unitaries with Pauli sparsity $s$, an additional challenge arises: the learning algorithms produce an operator $\hat{A}$ approximating $U$, which must be rounded to a nearby unitary $\hat{U}$. 
In general, it is unclear how to perform this rounding while preserving sparsity efficiently, and we leave this as an open problem. 
Nevertheless, in \cref{sec:sparsity}, we identify several natural settings in which this rounding step can be implemented efficiently. 
All circuit classes considered in this work fall into these settings, and hence all of our resulting learning algorithms run in polynomial time.\footnote{If one only needs to apply a \emph{quantum channel} close to the unknown unitary, the classical rounding step can be avoided. 
In particular, given $\hat{A}$, one can approximately implement its polar decomposition using standard techniques, e.g., the the algorithm of Quek and Rebentrost~\cite{quek2022fast}. This is because $\hat{A}$ is a linear combination of at most $s$ Pauli operators, and thus admits an efficient block-encoding via the Linear Combination of Unitaries (LCU) framework. See \cref{remark:quantum-polar} for details.}

Except for an edge case, the algorithm for learning low-Pauli-dimensional unitaries (\cref{thm:intro-pauli-dimension-sparsity}) is query-optimal.\footnote{A Pauli subgroup of order $2^k$ admits a canonical generating set $\{x_1,\dots,x_a,z_1,\dots,z_{a+b}\}$ with $2a+b=k$. The additional factor of $k$ in the $O(2^k k/\eps)$ query complexity of \cref{thm:intro-pauli-dimension-sparsity} arises only in the regime $a = o(\log b)$. In all other regimes, our algorithm is query optimal. Whether this factor is necessary remains open.}
When $k = 2n$, so that the Pauli spectrum has full support, \cref{thm:intro-pauli-dimension-sparsity} recovers the $\Theta(4^n/\eps)$ bound of Haah, Kothari, O’Donnell, and Tang~\cite{haah2023query}. 
Thus, our result generalizes their optimal tomography algorithm to a broader setting, interpolating between the worst-case regime of arbitrary unitaries and structured subclasses that admit faster learning.

Our work fits into a broader theme in quantum estimation: exploiting structure in the Pauli spectrum of quantum states and circuits. 
This perspective has led to efficient algorithms for learning and testing quantum states, channels, and unitary processes~\cite{gross2021schur,grewal2023efficient,grewal_et_al:LIPIcs.ITCS.2023.64,grewal2024agnostictomographystabilizerproduct,grewal2023improved,hangleiter2023bell,leone2023learning,chen2024stabilizerbootstrappingrecipeefficient,arunachalam2025testing,bao2025tolerant,hinsche2025testing,mehraban2025improved,iyer2025toleranttestingstabilizerstates,Flammia2021paulierror,bao2023testing,nadimpalli2024qac0,chen2023testing,arunachalam_et_al:LIPIcs.ICALP.2024.13,pmlr-v291-vasconcelos25a}, and for quantum estimation more generally~\cite{gu2025magicinduced,grewal2024pseudoentanglementaintcheap}.
Our results extend this paradigm to a unified framework for learning structured quantum circuits.

Finally, a recent work due to Honjani and Heidari~\cite{honjani2026querylearningnearlypauli} study the problem of learning an unknown unitary that is promised to be $\eps$-close to an $s$-Pauli-sparse unitary, where closeness is measured in the $\ell_1$-distance of the Pauli coefficients, i.e., two unitaries $U = \sum_P \alpha_P P$ and $V = \sum_P \beta_P P$ are $\eps$-close if $\sum_P \abs{\alpha_P - \beta_P} \le \eps$. 
Their algorithm uses $\widetilde{O}(s^6/\varepsilon^4)$ queries. In contrast, our \cref{thm:intro-pauli-dimension-sparsity} gives more efficient guarantees in the setting of exactly $s$-Pauli-sparse unitary channels.

\subsection{Applications}
\label{subsec:applications}

\cref{thm:intro-framework} yields efficient learning algorithms for a broad range of quantum circuit classes. 
At first glance, it may be unclear which circuits satisfy the condition that there exists a generating set $G$ for which $\{U^\dagger g U : g \in G\}$ have low-complexity Pauli spectra. 
We show that this condition in fact captures a wide variety of natural and well-studied circuit classes, providing a unifying explanation for many previously disparate learning results. 
Moreover, it extends beyond prior work, yielding efficient algorithms for more expressive classes of circuits than were previously known to be learnable.

We illustrate this with representative examples and defer full details to \cref{sec:applications}. 
In particular, \cref{thm:intro-framework} recovers prior results on learning arbitrary unitaries~\cite{haah2023query}, the Clifford hierarchy~\cite{low2009learning}, near-Clifford circuits~\cite{lai2022learning,leone-stabilizer-nullity}, shallow circuits~\cite{huang2024shallow}, fermionic matchgate circuits~\cite{oszmaniec2022fermion}, the matchgate hierarchy~\cite{matchgate-hierarchy}, and quantum juntas~\cite{chen2023testing}. 
Here, ``near-Clifford'' refers to circuits consisting of Clifford gates together with $O(\log n)$ non-Clifford single-qubit gates, and a quantum $k$-junta is a unitary acting nontrivially on only $k$ of $n$ qubits.

Importantly, prior works relied on techniques tailored to each circuit class. 
In contrast, all of these results arise as direct consequences of \cref{thm:intro-framework}, showing that the low-complexity Pauli spectra condition provides a unifying principle for efficient unitary learning.

\paragraph{Quantum $k$-juntas}\label{para:juntas-intros}
An $n$-qubit quantum $k$-junta is a unitary that acts nontrivially on only $k$ of the $n$ qubits. 
We obtain a query-optimal learning algorithm for quantum $k$-juntas.

\begin{corollary}[Informal version of \cref{cor:optimal-junta,thm:junta-lowerbound}]
\label{thm:junta-informal}
Let $U$ be an $n$-qubit quantum $k$-junta. 
Given query access to $U$, there is a time-efficient algorithm that outputs $V$ satisfying $\distdiamond(U,V) \leq \eps$ with probability at least $1-\delta$ using $O\!\left(\tfrac{4^k}{\eps}\log\!\tfrac{1}{\delta}\right)$ queries. 
This query complexity is optimal up to constant factors.
\end{corollary}

Our result yields exponential improvements in both query and time complexity over prior work. 
Compared to \cite{chen2023testing}, we improve the query complexity quadratically in $\eps$ while achieving the stronger guarantee of learning in diamond distance. 
The algorithm of \cite{huang2024shallow} relies on enumerating all $k$-local Pauli operators and is therefore efficient only when $k = O(1)$, whereas our algorithm remains efficient for $k = O(\log n)$.

\paragraph{Compositions of shallow and near-Clifford circuits}
Prior work gives efficient learning algorithms for constant-depth circuits~\cite{huang2024shallow} and for Clifford circuits with a small number of non-Clifford gates~\cite{leone-stabilizer-nullity}, but these techniques do not extend to their composition.

Using \cref{thm:intro-framework}, we obtain a unified algorithm that improves on both prior works and extends to compositions of shallow and near-Clifford circuits.

\begin{theorem}[Informal version of \cref{thm:learn-magic-heirarchy}]
\label{thm:learn-magic-heirarchy-informal}
Let $U$ be an $n$-qubit unitary of the form $U = QC$ or $U = CQ$, where $Q$ is a depth-$d$ circuit and $C$ is a Clifford circuit augmented with $t$ single-qubit non-Clifford gates. 
Then $U$ can be learned efficiently for $d = O(\log \log n)$ and $t = O(\log n)$.
\end{theorem}

This is the first result that simultaneously handles shallow and near-Clifford circuits, unifying the techniques developed for these settings.
Our algorithm is efficient for $d = O(\log \log n)$ and $t = O(\log n)$, which are both provably optimal.
In particular, we show that any algorithm requires exponential dependence on $t$ and doubly exponential dependence on $d$ (see \cref{sec:lower-bounds}).

Our algorithm remains efficient for depth $O(\log \log n)$ circuits, whereas \cite{huang2024shallow} can only learn constant-depth circuits. 
Similarly, \cite{leone-stabilizer-nullity} handles only Clifford circuits augmented with $T$ gates, whereas our algorithm applies to arbitrary single-qubit gates. 
We note that the class of unitaries covered by \cref{thm:learn-magic-heirarchy-informal} includes the first level of the recently introduced \emph{Magic Hierarchy}~\cite{parham2025quantumcircuitlowerbounds}.

\paragraph{Compositions of fermionic matchgates and Clifford circuits.}
Fermionic matchgates correspond to fermionic Gaussian unitaries under the Jordan--Wigner transformation~\cite{valiant2002quantum,terhal2002classical}. 
Prior work gives efficient learning algorithms for fermionic matchgate circuits~\cite{oszmaniec2022fermion,christensen2026learningfermioniclinearoptics}.

Using \cref{thm:intro-framework}, we extend these results to compositions with Clifford circuits, analogous to \cref{thm:learn-magic-heirarchy-informal}.

\begin{theorem}[Informal version of \cref{thm:learn-match-plus-clifford}]
\label{thm:intro-learn-matchgate-compositions}
Let $U$ be an $n$-qubit unitary of the form $U = FC$ or $U = CF$, where $F$ is a fermionic matchgate circuit and $C$ is a Clifford circuit. 
Then $U$ can be learned efficiently using queries $O(n^7/\eps^2)$ and time $O(n^8/\eps^2)$.
\end{theorem}

This shows that our framework extends beyond shallow circuits and juntas to other structured models such as fermionic systems, and continues to apply even under composition with Clifford operations.

\paragraph{Further applications}
We conclude with two brief remarks highlighting additional consequences of our framework.
First, our framework applies to a substantially more general class of unitaries than quantum $k$-juntas. For example, we can efficiently learn any unitary consisting of $O(1)$ alternating layers of near-Clifford circuits and arbitrary $O(\log n)$-juntas (not necessarily acting on the same qubits); see \cref{cor:alternating-junta-clifford}.

Second, our techniques naturally extend to alternative distance measures. 
For example, one can obtain analogues of our results in Frobenius distance, recovering and extending prior work on learning $\mathsf{QAC}^0$ circuits~\cite{pmlr-v291-vasconcelos25a} and low-degree unitaries~\cite{arunachalam_et_al:LIPIcs.ICALP.2024.13}. 
We omit the details, as they closely follow from our framework and do not introduce new technical ideas.

\subsection{Technical Overview}

\paragraph{The unifying framework}
At a high level, our approach reduces the problem of learning an unknown unitary $U$ to learning the Heisenberg evolution of a generating set of Pauli operators. 
Concretely, \cref{thm:intro-framework} shows that it suffices to learn approximations to the operators $\{U^\dagger g U : g \in G\}$.
This reduction proceeds in three steps.

First, in \cref{thm:generating-set-suffices}, we establish an information-theoretic guarantee: if two unitaries $U$ and $V$ have similar conjugation action on a generating set $G$ (i.e., $U^\dagger g U$ and $V^\dagger g V$ are close for all $g\in G$), then $U$ and $V$ are close in diamond distance. 
This shows that learning the Heisenberg evolution of $G$ is information-theoretically sufficient to learn $U$ itself.

Second, we give an efficient \emph{compiler} (\cref{thm:efficient-learning-thm}) that reconstructs a unitary $V$ from approximate descriptions of $\{U^\dagger g U : g \in G\}$. 
Our construction builds on the approach of Huang et al.~\cite{huang2024shallow}, who showed how to combine information about weight-one Pauli operators into a global approximation of $U$.
We generalize this approach to arbitrary generating sets of Pauli operators.
A key idea in the construction is to express $U^\dagger \otimes U$ as a product of local terms, each depending only on the conjugation of single-qubit Pauli operators, allowing us to assemble a global approximation from a product of local estimates.

Third, we design learning algorithms for the conjugated operators $U^\dagger g U$. 
Specifically, we design algorithms to learn the conjugated operators when they have low-complexity Pauli spectra. 
In particular, we develop algorithms for the cases where the operators are low Pauli dimensional and Pauli sparse.

\paragraph{Learning $k$-Pauli dimensionality}\label{para:k-pauli-dimension}
We now sketch our algorithm for learning unitaries whose Pauli spectrum is supported on a subgroup $G$ of size $2^k$.

At a high level, the algorithm proceeds in four steps. 
First, we approximately learn the subgroup $G$ supporting the Pauli spectrum of $U$ via Bell sampling of the Choi state. 
Second, we conjugate $U$ by an efficiently computable Clifford circuit to reduce the problem to learning a unitary that is \emph{approximately block-diagonal}. 
Third, we learn this approximately block-diagonal unitary and reconstruct $U$. 
Finally, we apply a bootstrapping technique to improve the dependence on $\eps$ from $1/\eps^2$ to $1/\eps$, achieving Heisenberg-limited scaling without inverse access to $U$.

The main technical challenge lies in the third step. 
A natural approach is to apply the unitary learning algorithm of \cite{haah2023query} independently to each block. 
However, this fails to preserve relative phase information between blocks and leads to an overall query complexity of $O(4^k)$, which is considerably worse than the $O(2^k)$ achieved by our algorithm.
Our key idea is to instead learn all blocks \emph{simultaneously}. 
We perform tomography on superpositions over the block index so that each recovered column encodes all blocks in parallel, preserving their relative phases.

This parallelization approach, however, only succeeds if the unitary is \emph{exactly} block-diagonal. 
Our reduction yields only an \emph{approximately} block-diagonal unitary, and the resulting off-block entries introduce noise that destroys the parallelization.
A key technical tool developed in this work is \emph{Pauli projection}. 
Given a subgroup $\widehat{G}$ and query access to $U=\sum_{P\in\mathcal P}\alpha_P P$, we define the Pauli projection of $U$ onto $\widehat{G}$ as
\[
\Pi_{\widehat{G}}(U) \coloneqq \sum_{P\in \widehat{G}} \alpha_P P.
\]
Note that $\Pi_{\widehat{G}}(U)$ will not generally be unitary. 
The crucial observation is that $\Pi_{\widehat{G}}(U)$ admits a Pauli-twirl representation,
\[
\Pi_{\widehat{G}}(U)=\mathbb E_{P\in \widehat{G}^{\sympcomp}}[PUP],
\]
which allows us to implement a block-encoding of $\Pi_{\widehat{G}}(U)$ using only $O(1)$ forward queries to $U$. 
This enables us to replace an approximately block-diagonal unitary with an exactly block-diagonal proxy that remains close to the original unitary, allowing the parallel tomography procedure to succeed.
To our knowledge, this is the first use of linear-combination-of-unitaries (LCU) and block-encoding techniques to enable a unitary or state learning algorithm.

\paragraph{Learning $s$-Pauli sparsity}
We now sketch our algorithm for learning unitaries whose Pauli spectrum is supported on $s$ operators. 
At a high level, we reduce unitary learning to a sparse state tomography problem.

The key observation is that the Choi state of a unitary $U = \sum_{P \in \supp(U)} \alpha_P P$ can be written as a superposition over $s$ Bell basis states, with amplitudes given by the Pauli coefficients $\{\alpha_P\}$. 
Thus, learning $U$ reduces to the following task: given copies of a state $\ket{\psi}$ supported on $s$ computational basis states, output a classical description of $\ket{\psi}$ that is $\eps$-close in trace distance.

We develop a copy-optimal tomography algorithm for such states, which uses $O(s/\eps^2)$ samples and runs in $\poly(s)$ time.  
By running the state tomography procedure, we can learn approximations $\{\wh{\alpha}_x\}$ of the Pauli coefficients and output the operator $\hat{A} = \sum_{P} \wh{\alpha}_P P$ that approximates $U$ and has support contained in $\supp(U)$.

A key distinction from the low-Pauli-dimensional setting is that this approach yields an \emph{improper} learner in general: the output $\hat{A}$ is not guaranteed to be unitary. 
While information-theoretically one can round $\hat{A}$ to a nearby unitary (e.g., via polar decomposition), doing so efficiently while preserving sparsity appears challenging, and whether such rounding can be performed in general remains an open question.
Nevertheless, we identify several natural settings in which efficient rounding is possible, and all circuit classes studied in this work fall into these settings. 
Consequently, all of our resulting learning algorithms run in polynomial time.

\paragraph{Applications}
The applications highlighted in \cref{subsec:applications} are obtained by instantiating our framework on specific circuit classes, by showing that their conjugation action on a suitable generating set has low-complexity Pauli spectra. 
We further extend this approach to an infinite hierarchy of unitary classes in \cref{subsec:hierarchy}, yielding a general framework that also encompasses the Clifford hierarchy~\cite{low2009learning}, the matchgate hierarchy~\cite{matchgate-hierarchy}, and compositions such as $O(1)$ alternations of $O(\log n)$-juntas with near-Clifford circuits. We defer further details to \cref{sec:qnc-plus-clifford,sec:applications}.

\subsection{Open Problems}
\label{subsec:disc-open-prob}

Our work takes a step toward unifying unitary learning algorithms and extending them to more expressive circuit classes. 
It leaves open several interesting problems related to the efficient learnability of structured quantum circuits.

\paragraph{A framework for state learning}
In \cref{thm:intro-framework}, we presented a sufficient condition for efficient unitary learning. 
A natural question is whether there is an analogous framework for learning quantum states. 
Is there a structural condition---analogous to the low-complexity Pauli spectrum condition studied here---that characterizes efficiently learnable classes of quantum states?

\paragraph{Extending \cref{thm:intro-framework}}
\cref{thm:intro-framework} reduces learning a unitary to learning the conjugated images of a generating set of Pauli operators, using a \emph{nonadaptive} reduction. 
Can this reduction be strengthened further? 
For example, can adaptive strategies extend the scope of the framework or improve its efficiency?

It would also be valuable to identify additional natural and physically motivated classes of unitaries that satisfy the condition in \cref{thm:intro-framework}. 
Conversely, it would be equally interesting to understand the limitations of the framework by exhibiting efficiently learnable classes of unitaries that do not appear to fit within it.

\paragraph{Inverse-free learning and pseudorandomness.}
The algorithms obtained through \cref{thm:intro-framework} require query access to both $U$ and $U^\dagger$, whereas the algorithms in \cref{thm:intro-pauli-dimension-sparsity,thm:junta-informal} use only forward access to $U$. 
For instance, \cref{thm:learn-magic-heirarchy-informal}, which learns compositions of depth-$O(\log \log n)$ circuits with near-Clifford circuits, uses inverse queries. 
Can this class be learned using only forward access to $U$?

Both possibilities would be interesting. 
A positive answer would require new algorithmic ideas beyond those developed here. 
A negative answer could have implications for quantum pseudorandomness: it would suggest a class of unitaries that is learnable given access to both $U$ and $U^\dagger$, but not with access to $U$ alone. 
Such a class would provide a candidate separation between weak and strong pseudorandom unitary ensembles (PRUs), where weak PRUs are secure against forward access while strong PRUs remain secure even given access to $U^\dagger$.

\paragraph{Rounding of Pauli-sparse matrices to unitaries}
Our learning algorithm for $s$-sparse unitaries outputs an approximate matrix $\hat{A}$ that is close to a Pauli-sparse unitary $U$ using polynomial time and queries. 
However, we do not know how to efficiently round $\hat{A}$ to a unitary.
We identify several natural settings where such rounding can be performed efficiently (see \cref{fact:very-sparse-rounding,fact:commute-rounding,fact:majorana-rounding}), and all applications in this work fall into these settings.
Ideally, Pauli sparsity is also preserved in the rounding process.
Whether efficient sparsity-preserving rounding is possible in general remains an interesting open problem.

\section{Preliminaries}\label{sec:prelims}

We use $\exp(x) = e^x$, and $\log$ denotes the natural logarithm.
For a set $S = \{x_1, \dots, x_k\} \subseteq \F^n$, we write $\langle S \rangle$ for the linear span of the elements of $S$, and $\dim(\cdot)$ for the dimension of a subspace. 

We will use the following standard consequence of a multiplicative Chernoff bound, which ensures that, with high probability, a sufficient number of successes occur in repeated Bernoulli trials.

\begin{lemma}\label{lemma:bernoulli_bound_mult}
	Let $X_1, \dots, X_m$ be i.i.d.\ Bernoulli random variables with probability $p$. If \[m = \frac{2}{p}\left(d + \log(1/\delta)\right),\] then
	\[
	\Pr\left[\sum_{i=1}^m X_i < d \right] \leq \delta.
	\]
\end{lemma}

Let $\calP_n \coloneqq \{\alpha P_1 \otimes \dots \otimes P_n : P_i \in \{I, X, Y, Z\}, \,\alpha \in \{\pm 1, \pm i\}\}$ denote the $n$-qubit Pauli group.
It is often convenient to identify Pauli operators (up to phase) with elements of $\F_2^{2n}$ (and vice versa).
In particular, $x = (a,b) \in \F_2^{2n}$ corresponds to the Pauli operator 
\[
W_x = i^{a \dot b} \bigotimes_{j=1}^n X^{a_j} Z^{b_j} \in \calP_n.
\]
We refer to the set $\{W_x\}_{x \in \F_2^{2n}} \subseteq \calP_n$ as the \emph{Weyl operators}, i.e., the set of Pauli operators modulo phase.
Formally, consider $\Pmod \coloneqq \calP_n / \{\pm 1, \pm i\}$, the Pauli group modulo phase; each coset in $\Pmod$ is represented uniquely by some $W_x$. 

We equip the space $\F_2^{2n}$ with the following bilinear form, called the symplectic product.
\begin{definition}[Symplectic product]
For $x,y \in \F_2^{2n}$, the \emph{symplectic product} is defined as
\[
[x,y] \coloneqq \sum_{i=1}^n x_i y_{n+i} \;+\; x_{n+i} y_i \quad (\mathrm{mod}\; 2).
\]
\end{definition}

The symplectic product precisely captures commutation relations: $[x,y] = 0$ if and only if $W_x$ and $W_y$ commute, and $[x,y]=1$ if and only if they anticommute.

The symplectic product gives rise to the notion of a symplectic complement. 

\begin{definition}[Symplectic complement]
    For a subspace $S \subseteq \F_2^{2n}$, the symplectic complement $S^\sympcomp$ is defined as
    \[
        S^\sympcomp \coloneqq \{x \in \F_2^{2n} : \forall y \in S,\, [x,y] = 0\}.
    \]
\end{definition}

Equivalently, if $S$ is viewed as a set of Weyl operators, then $S^\sympcomp$ is the set of all Weyl operators that commute with every element of $S$. 
The symplectic complement is always a subspace, and satisfies the dimension identity $\dim(S) + \dim(S^\sympcomp) = 2n$. 

We will also need the following basic Fourier-analytic fact, which says that symplectic characters of $S^\sympcomp$ vanish unless the input lies in $S$.

\begin{proposition}[{See e.g., \cite[Lemma 2.11]{grewal2023improved}}]\label{prop:sum-over-characters}
    For all subspaces $S \subseteq \F_2^{2n}$:
    \[\E_{x \in S^\sympcomp}\left[(-1)^{[a, x]}\right] = \mathbbm{1}_{a \in S}\]
    where $\mathbbm{1}$ is the indicator function.
\end{proposition}

The Weyl operators form an orthonormal basis under the normalized Hilbert-Schmidt inner product $\langle A,B\rangle=\tfrac1{2^n}\Tr(A^\dagger B)$, so every operator admits a unique expansion in this basis. 

\begin{definition}[Pauli expansion]
The \emph{Pauli expansion} of $A \in \C^{2^n \times 2^n}$ is 
\[
A = \frac{1}{2^n} \sum_{x \in \F_2^{2n}} \tr(AW_x)\, W_x.
\]
The coefficients $\tr(AW_x)$ are called the \emph{Pauli spectrum} of $A$.
\end{definition}

We introduce the following notation to track the Weyl operators that appear in the expansion.

\begin{definition}[Pauli support]
The \emph{Pauli support} of $A \in \C^{2^n \times 2^n}$ is defined as 
\[
\supp(A) \coloneqq \{x \in \F_2^{2n} : \tr(W_x A) \neq 0\}.
\]
\end{definition}

We now introduce Pauli dimensionality, a notion of complexity defined via the Pauli expansion of an operator. It can be viewed as the noncommutative analogue of Fourier dimensionality from Boolean analysis~\cite{gopalan2011testing}.

\begin{definition}[Pauli dimensionality]
A matrix $A$ has Pauli dimension $k$ (or, is $k$-Pauli-dimensional) when the span of its support, $\langle \supp(A) \rangle$, has dimension $k$.
\end{definition}

Since all supports, expansions, and spans in this work are taken with respect to the Pauli/Weyl basis, we often omit the qualifier ``Pauli'' and simply speak of, e.g., the ``support of $U$'' or the ``span of $\{W_1,\dots,W_m\}$.''

The following subset of Weyl operators will appear often in our analysis.
\begin{definition}\label{def:pauli-support-block-diagonal}
For $a, b, n \in \N$ with $a + b \leq n$, define
\[
\paulisupport_{a,b} \coloneqq 
0^{\,n-a} \times \F_2^a \times 0^{\,n-a-b} \times \F_2^{\,a+b}.
\]
As operators, this subspace corresponds to
\[
I^{\otimes (n-a-b)} \otimes \{I,Z\}^{\otimes b} \otimes \{I,X,Y,Z\}^{\otimes a}.
\]
In words, $\paulisupport_{a,b}$ consists of all $n$-qubit Pauli operators that act trivially on the first $n-a-b$ qubits, act as either $I$ or $Z$ on the next $b$ qubits, and act arbitrarily on the final $a$ qubits.
\end{definition}

We conclude this section by establishing our notation and conventions for distances between unitary channels. 
We define the operator norm $\opnorm{A} \coloneqq \max_i \sigma_i$, the Frobenius norm $\fnorm{A} \coloneqq \sqrt{\sum_i \sigma_i^2}$, and the trace norm $\norm{A}_{\tr} \coloneqq \sum_i \sigma_i$, where $\{\sigma_i\}$ are the singular values of $A$.
These are all Schatten norms, and hence unitarily invariant: $\norm{UAV} = \norm{A}$ for all unitaries $U,V$.
These norms induce distances $\dist(U,V) = \norm{U-V}$ between unitaries.
We record several standard facts:

\begin{fact}\label{fact:-weyl-to-frob}
    For matrix $A \in \C^{2^n \times 2^n}$ with Weyl decomposition $A = \sum_{x \in \F_2^{2n}} \alpha_x W_x$, \[\fnorm{A} = \sqrt{2^n  \sum_{x \in \F_2^{2n}} \abs{\alpha_x}^2} = \sqrt{\sum_{x, y \in \F_2^n} \abs{\braket{x | A | y}}^2}.\]
\end{fact}

\begin{fact}\label{fact:operator-induced}
    For matrix $A \in \C^{d \times d}$, \[\opnorm{A} = \norm{\max_{\ket \psi} A\ket{\psi}}_2 = \max_{\ket \phi,\, \ket \psi} \braket{\phi | A | \psi},\]
    where the maximization is over unit vectors $\ket \psi$ and $\ket \phi$.
\end{fact}

\begin{fact}\label{fact:schatten-norm-conversion}
    For matrix $A \in \C^{d \times d}$, with rank $r$, $\opnorm{A} \leq \fnorm{A} \leq \norm{A}_\tr \leq \sqrt{r}\fnorm{A} \leq r \opnorm{A}$.
\end{fact}

\begin{fact}\label{fact:frob-repeat-identity}
    Let $I_d$ be the $d \times d$ identity matrix, then $\fnorm{I_d \otimes A} = \sqrt{d} \cdot \fnorm{A}$ and $\opnorm{I_d \otimes A} = \opnorm{A}$.
\end{fact}

Next, we define the diamond distance, which measures the maximum distinguishability between two unitary channels, even in the presence of ancilla.\footnote{The diamond distance is defined more generally for arbitrary quantum channels, but we restrict to unitary channels since that is the only case needed in this work.}

\begin{definition}[Diamond distance]
For unitaries $U,V \in \C^{d \times d}$,
\begin{equation*}
    \distdiamond(U, V) \coloneqq \max_\rho \norm{(I_A \otimes U) \rho(I_A \otimes U^\dagger) - (I_A \otimes V) \rho(I_A \otimes V^\dagger)}_\tr
\end{equation*}
where $I_A$ is the identity operator on the ancilla register.
\end{definition}

It is often easier to work with the following distance measure, which implies bounds on the diamond distance. 

\begin{definition}[Phase-Aligned Operator Distance]\label{def:phase-op-dist}
    For unitaries $U, V \in \C^{d \times d}$,
    \[
        \distphop(U, V) = \min_{\theta \in [0, 2\pi)} \opnorm{e^{i \theta} U - V}.
    \]
\end{definition}

\begin{fact}[{\cite[Proposition 1.6]{haah2023query}}]\label{fact:op-to-diamond}
    For all unitaries $U, V \in \C^{d \times d}$,
    \[
    \distdiamond(U, V) \leq 2 \distphop(U, V) \leq 2\distdiamond(U, V).
    \]
\end{fact}

We note that $\distdiamond$ and $\distphop$ are both unitarily invariant. 

\begin{fact}\label{fact:dist-composition}
    For unitaries $U_1, U_2, V_1 V_2 \in \C^{d \times d}$ and any unitarily invariant distance $\dist(\cdot, \cdot)$:
    \[\dist(U_1 U_2, V_1 V_2) \leq \dist(U_1, V_1) + \dist(U_2, V_2).\]
\end{fact}
\begin{proof}
We have
    \[
        \dist(U_1 U_2, V_1 V_2) \leq \dist(U_1 U_2, V_1 U_2) + \dist(V_1 U_2, V_1 V_2) = \dist(U_1, V_1) + \dist(U_2, V_2),
    \]
    where the first inequality follows from the triangle inequality and the second follows from the fact that $\dist(\cdot, \cdot)$ is unitarily invariant. 
\end{proof}

Finally, we mention a second important distance measure used in prior work (see \cite[Section 5.1.1]{montanaro2016survey} for details).
It can be seen as a more average case distance as $\dist_{\text{phaseF}}(U, V) \leq \distphop(U, V) \leq \sqrt{d} \cdot  \dist_{\text{phaseF}}(U, V)$.
It is also unitarily invariant.

\begin{definition}[Phase-aligned normalized Frobenius distance]\label{def:frob-dist}
    For unitaries $U, V \in \C^{d \times d}$,
    \[
        \dist_{\text{phaseF}}(U, V) = \frac{1}{\sqrt{d}} \min_{\theta \in [0, 2\pi)} \fnorm{e^{i \theta} U - V}.
    \]
\end{definition}

\section{Pauli Projection via Linear Combination of Unitaries}\label{sec:lcu}

In this section, we introduce Pauli projection, which lets us ``filter out'' certain Pauli operators from the spectrum of a unitary. 
Specifically, given a subspace of Pauli operators $S$ and query access to a unitary $U$, our goal is to apply the operator obtained from $U$ by zeroing out all $U$'s Pauli coefficients that are outside of $S$. 
The procedure relies on block encodings and the linear combination of unitaries technique. 
Although these tools are most commonly used for Hamiltonian simulation and related linear-algebraic tasks, we leverage them in a novel way to design algorithms for learning unitary channels. 

We refer to the projected operator---obtained by restricting $U$ to its Pauli coefficients inside $S$---as the \emph{Pauli projection}, which we now formalize in the following definition.

\begin{definition}[Pauli projection]\label{def:pauli-proj}
Let $a, b, n \in \N$ with $a + b \leq n$, and let $A \in \C^{2^n \times 2^n}$ be a matrix. Given a subspace $S \subseteq \F_2^{2n}$, the Pauli projection of $A$ onto subspace $S$, denoted $\Pi_S(A)$, is defined as 
\[
    \Pi_S(A) \coloneqq \frac{1}{2^n}\sum_{x \in S} \tr(A W_x) W_x.
\]
In other words, we zero out all Pauli coefficients of $A$ that are outside of $S$.
\end{definition}

The Pauli projection of a unitary channel can be expressed as a convex combination of unitaries. In fact, it can be written as a Pauli twirl.

\begin{lemma}\label{lemma:pauli-twirl}
Let $U \in \C^{2^n \times 2^n}$ be a unitary matrix, and let $S \subseteq \F_2^{2n}$ be a subspace. The Pauli projection of $U$ onto $S$ can be expressed as a Pauli twirl:
\[
\Pi_S(U) = \E_{q \sim S^\perp}[W_q U W_q^\dagger], 
\]
where the Pauli operator $W_q$ is chosen uniformly at random from $S^\perp$.
\end{lemma}
\begin{proof}
Let $W_x$ be an arbitrary Weyl operator.
We have 
\begin{align}
    \E_{q \sim S^\perp}[W_q W_x W_q^\dagger]
    &= \E_{q \sim S^\perp}[(-1)^{[q,x]}] W_x. 
\end{align}
If $x \in S$, then $\E_q[(-1)^{[q,x]}] = 1$. 
However, if $x \notin S$, then $\E_q[(-1)^{[q,x]}] = 0$ by \cref{prop:sum-over-characters}. Hence, by linearity, we have that, for any operator $A$, the Pauli twirl will zero out all $W_x \notin S$, which is precisely the action of $\Pi_S(\cdot)$.
\end{proof}

Next, we recall a standard tool in quantum algorithms: block encodings.

\begin{definition}[Block encoding {\cite[Definition 1.1]{tang2023quantum}} (See also {\cite[Definition 43]{gilyen2019qsvt}},{\cite[Definition 1]{rall2020physical-quantities}}]
   Let $A \in \C^{r \times c}$, and let $U \in \C^{2^n \times 2^n}$ be a unitary matrix. We say that $U$ is a $Q$-\emph{block encoding} of $A$ if $U$ is implementable with $O(Q)$ gates and has the form 
   \[
   U = \begin{pmatrix}
   A & \cdot \\ 
   \cdot & \cdot 
   \end{pmatrix},
   \]
   where $\cdot$ denotes unspecified block matrices.
\end{definition}

\begin{lemma}[{\cite[Lemma 1.5]{tang2023quantum}}]\label{lemma:block-encoding-state-prep}
   Let $U \in \C^{2^n \times 2^n}$ be a $Q$-block encoding of $A \in \C^{r \times c}$, and let $\ket\psi \in \C^c$ be an input state. Then there is a quantum circuit using $1$ query to $U$ (and therefore $O(Q)$ gates) that prepares the state $\frac{A \ket\psi}{\norm{A\ket\psi}_2}$ with probability $\norm{A \ket\psi}_2^2$.
\end{lemma}

We use the linear combination of unitaries (LCU) method to construct block encodings of more general operators. 
This technique allows one to turn a weighted sum of unitaries into a block encoding of the corresponding operator.

\begin{lemma}[{\cite{childs-kothari-somma2017,gilyen2019qsvt}}]\label{lemma:lcu}
 Let $A = \sum_k a_k U_k$ be a linear combination of unitary operators where $a_k \geq 0$. Define the preparation operator $\PREP$ as
 \[
 \PREP \ket{0} = \sum_k \sqrt{\frac{a_k}{\sum_k a_k}} \ket{k}, 
 \]
 and the select operator $\SEL$ as 
 \[
 \SEL \ket k \ket \psi = \ket k U_k \ket\psi. 
 \]
 Then the circuit $\PREP^\dagger \cdot \SEL \cdot \PREP$ is a block encoding of $\frac{A}{\sum_k a_k}$.
\end{lemma}

We are now ready to present the main algorithm of this section. In particular, given a single query to $U$ and a subspace of Pauli operators $S$, one can perform the mapping \[\ket\psi \mapsto \frac{\Pi_S(U)\ket\psi}{ \norm{\Pi_S(U)\ket\psi}_2}\] with postselection.

\begin{lemma}\label{lemma:apply-pauli-projection}
   Let $U \in \C^{2^n \times 2^n}$ be a unitary matrix, let $S\subseteq \F_2^{2n}$ be a subspace, and let $\ket\psi$ be an $n$-qubit input state. There is a quantum algorithm that uses a single query to $U$ and prepares the state $\frac{\Pi_S(U)\ket\psi}{\norm{\Pi_S(U)\ket\psi}_2}$ with probability $\norm{\Pi_S(U)\ket\psi}_2^2$. 
   The algorithm uses $O(n \cdot \dim(S^\sympcomp)) = O(n^2)$ gates and requires $\dim(S^\perp) \leq 2n$ ancilla qubits.
\end{lemma}
\begin{proof}
    By \cref{lemma:pauli-twirl}, we can write $\Pi_S(U) = \E_{q \sim S^\perp}[W_q U W_q^\dagger]$, where the expectation is over uniformly random $q \in S^\perp \subseteq \F_2^{2n}$.
    Note that this is a convex combination of unitaries (i.e., the sum of the coefficients is $1$, so we don't run into scaling issues from \cref{lemma:lcu}). 
    Let $\ell = \dim S^\perp$, and identify each operator in $S^\perp$ with an index $k \in [2^\ell]$. 
    Let $\PREP$ denote a unitary that maps \( \ket{0^\ell} \) to the uniform superposition, i.e., Hadamard gates on $\ell$ qubits. 
    Let the select operator $\SEL$ be $\ketbra{k}{k}\otimes W_k U W_k^\dagger$. Note that $\SEL$ can be decomposed as $(\ketbra{k}{k} \otimes W_k) \cdot (I^{\otimes \ell} \otimes U)\cdot (\ketbra{k}{k}\otimes W_k^\dagger)$, and hence can be implemented using two controlled-Pauli operations and a single query to $U$.
    Then, by \cref{lemma:lcu}, $\PREP^\dagger \cdot \SEL \cdot \PREP$ is a block encoding of $\Pi_S(U)$. Finally \cref{lemma:block-encoding-state-prep} implies that we can prepare $\frac{\Pi_S(U)\ket\psi}{\norm{\Pi_S(U)\ket\psi}_2}$ with probability $\norm{\Pi_S(U)\ket\psi}_2^2$.

    Constructing the unitary $\ketbra{k}{k} \otimes W_k$ can be done by first taking generators of $S^\sympcomp = \langle \{g_1, \dots, g_\ell\} \rangle$.
    Then have each $W_{g_i}$ be controlled on qubit $i$ of the control register.
    Each controlled $W_{g_i}$ can be constructed as the product of at most $2n$ controlled single-qubit Pauli operations (with the control, this becomes a two-qubit gate), for a total of at most $O(n \ell)$ two-qubit gates.\footnote{In the particular scenario that we will use \cref{lemma:apply-pauli-projection}, we can find a set of \emph{sparse} generators for $S^\sympcomp$ such that each controlled $W_{g_i}$ is already just a two-qubit gate. This leads to only needing $O(\dim(S^\sympcomp))$ many additional gates.}
\end{proof}

\section{Learning Approximately Block-Diagonal Unitaries}\label{sec:learn-block-diag}

In this section, we develop an algorithm for learning unitary channels that are approximately block-diagonal. 
The algorithm for learning $k$-Pauli-dimensional unitary channels, which we present in \cref{sec:complete-algo}, will follow from a reduction to the approximately block-diagonal setting. 
We begin by formally defining ``approximately block-diagonal.''

\begin{definition}[Block-diagonal matrix]
Let $a, b, n \in \N$ with $a + b \leq n$, and let $A \in \C^{2^n \times 2^n}$ be a matrix. 
We say that $A$ is $(a,b)$-block-diagonal if there exists $2^b$ matrices $A_y \in \C^{2^a \times 2^a}$ such that $A = I^{\otimes n-a-b} \otimes \left( \bigoplus_{y \in \{0,1\}^b}A_y \right)$.
That is, $A$ is a block-diagonal matrix whose diagonal blocks are $A_1$, \dots, $A_{2^b}$, each of size $2^a \times 2^a$, and this $2^{a+b} \times 2^{a+b}$ block structure is repeated $2^{n-a-b}$ times along the diagonal. 
\end{definition}

Let $A = I^{\otimes n-a-b} \otimes \left( \bigoplus_{y \in \{0,1\}^b}A_y \right)$ be an $(a,b)$-block-diagonal matrix.
It is helpful to view $A$ as acting on $n$ qubits, grouped into three registers: 
the first $n-a-b$ qubits are unaffected by $A$; the second $b$ qubits select which diagonal block we are in; and the last $a$ qubits index a column within that block. 
For $x \in \{0, 1\}^{n-a-b}$, $y \in \{0, 1\}^b$, and $z \in \{0, 1\}^a$, we have $A(\ket x \! \ket y \! \ket z) = \ket{x} \ket{y} (A_y \ket z)$.

\begin{definition}[Block-diagonal unitary]
A matrix $U \in \C^{2^n \times 2^n}$ is an $(a,b)$-block-diagonal unitary if it is both unitary and an $(a,b)$-block-diagonal matrix.
\end{definition}

\begin{definition}[Approximately block-diagonal unitary]
    Let $a,b,n \in \N$ with $a+b \leq n$, and let $U \in \C^{2^n \times 2^n}$ be a unitary matrix. We say $U$ is $(a,b,\eps)$-approximately block-diagonal (with respect to a given matrix norm $\norm{\cdot}$) if there exists an $(a,b)$-block-diagonal unitary $V$ satisfying $\norm{U-V} \leq \eps$. 
\end{definition}

Unless otherwise stated, $(a,b,\varepsilon)$-approximately block-diagonal unitaries are always with respect to the operator norm $\opnorm{\cdot}$.

With these definitions in place, we can now state our main technical result for this section: a tomography algorithm that efficiently estimates approximately block-diagonal unitaries in diamond norm.

\begin{restatable}{theorem}{approxblockdiag}\label{thm:approx-block-diag}
Let $a,b, n \in \N$ with $a + b \leq n$.
There is a tomography algorithm that, given query access to an $(a,b,\eps)$-approximately block diagonal unitary channel $U \in \C^{2^n \times 2^n}$ as well as parameters $\delta, \eps > 0$, applies $U$ at most $O\left(2^{a+b}\cdot (2^a+b)\frac{\log(1/\delta)}{\eps^2}\right)$ times and outputs an estimate $V$ satisfying $\distphop(U, V) \leq O(\eps)$ with probability at least $1-\delta$. 
Moreover:
\begin{itemize}
    \item The output $V$ is $(a,b)$-block-diagonal.
    \item The algorithm runs in time $\poly\!\left(n,\, 2^{2a+b},\, \varepsilon^{-1},\, \log \delta^{-1}\right)$.
    \item The algorithm uses only forward access to $U$ (i.e., it does not require $U^\dagger$ or controlled-$U$).
    \item The algorithm uses $2n - 2a -b$ additional qubits of space.
\end{itemize}
\end{restatable}

Parameter counting shows that $\Omega(2^{2a +b})$ queries are necessary, so our dependence on $a$ and $b$ is optimal up to an additive logarithmic factor in $b$.
Achieving this optimal scaling requires some care.
Ideally, we would like to simply treat $U$ as being truly block-diagonal and therefore learn all of the relevant blocks of $U$ in parallel (recall that learning them sequentially would require $O(4^{a+b})$ queries to recover the relative phases between blocks), as described in \cref{para:k-pauli-dimension}.
However, because $U$ is only $\eps$-close to block-diagonal, its off-diagonal blocks behave like noise, and a na\"ive parallel approach will fail.

To overcome this, we replace $U$ with a Pauli projection $\Pi_{\paulisupport_{a,b}}(U)$ (\cref{def:pauli-support-block-diagonal,def:pauli-proj}) that has three crucial properties: (i) it is exactly $(a,b)$-block diagonal, (ii) it is within $2\eps$ of $U$ in operator norm, and (iii) it can be efficiently applied using LCUs and block encodings. 
With this operator in hand, our algorithm essentially applies $\Pi_{\paulisupport_{a,b}}(U)$ to states of our choice and then uses state tomography to recover the columns of all blocks of $\Pi_{\paulisupport_{a,b}}(U)$ simultaneously, thereby achieving the optimal dependence on $a$ and $b$. 

We now present our learning algorithm.

\begin{algorithm}[H]
\SetKwInOut{Promise}{Promise}
\caption{Learning approximately block-diagonal unitaries}
\label{alg:base-approx-bd}
\DontPrintSemicolon
\KwInput{Query access to an $(a,b,\eps)$-approximately block diagonal unitary $U$ for $\eps \leq 1/K$ where $K \geq 1$ is some universal constant}
\KwOutput{Description of a $(a,b)$-block-diagonal unitary $V$ such that $\distphop(U, V) \leq O(\eps)$ with probability at least $\frac{2}{3}$}

\ForEach{$z \in \{0,1\}^a$}{
Use \cref{lemma:apply-pauli-projection} to apply $\Pi_{\paulisupport_{a,b}}(U)$ to the $n$-qubit state $\ket{0^{n-a-b}}\ket{+}^{\otimes b}\ket{z}$.
\label{alg-step:postselect}

Discard the first register to obtain an $(a+b)$-qubit state $\ket{\psi_z}$ on the remaining qubits. 

Run the pure state tomography algorithm described in \cref{lem:state-tomo} on $\ket{\psi_z}$ to error $\eps \cdot \sqrt{\frac{2^a}{2^a+b}}$ and failure probability $\exp(-5 \cdot 2^{a+b})$. Denote the output of this step by $\ket{\wh{\psi_z}}$.
}
\label{alg-step:tomography}

\ForEach{$y \in \{0, 1\}^b$}{
\label{alg-step:collate-0}

Construct a matrix $A_y \in \C^{2^a \times 2^a}$ whose $z$th column is $\sqrt{2^b} \cdot \left(\bra{y} \otimes I^{\otimes a}\right) \ket{\wh{\psi_z}}$.
\label{alg-step:collate-1}

Compute $A_y = L_y\Sigma_y R_y^\dagger$, the SVD of $A_y$. Let $V_y = L_y R_y^\dagger$.
}
\label{alg-step:collate-2}

Let $V = I^{\otimes n-a-b} \otimes \left(\bigoplus_{y \in \{0, 1\}^b} V_y\right)$.
\label{alg-step:merge-blocks}

Repeat the above procedure on $U (I^{\otimes n-a} \otimes \!H^{\otimes a})$ to recover $V' = I^{\otimes n-a-b} \otimes \!\left(\!\bigoplus_{y \in \{0, 1\}^b} V_y'\right)$.
\label{alg-step:fix-phase1}

Let $D \in \C^{2^a \times 2^a}$ be the diagonal unitary obtained from applying \cref{lem:phase-align-unitary} to $V_0$ and $V_0'$.
\label{alg-step:fix-phase3}

\Return $V (I^{\otimes n-a} \otimes D) = I^{\otimes n-a-b} \otimes \left(\bigoplus_{y \in \{0, 1\}^b} V_y D\right)$.
\label{alg-step:fix-phase5}
\end{algorithm}

The core of the algorithm lies in the first two stages.
The initial foreach loop performs tomography to learn (up to relative phase) the columns of every block of $\Pi_{\paulisupport_{a,b}}(U)$ in parallel, and the second foreach loop assembles these into block matrices and rounds each to the nearest unitary via polar decomposition (i.e., sets the singular values to be $1$).
The final stage approximately recovers the relative phases between the columns within each block; this phase-alignment procedure follows the approach of \cite[Proposition 2.3]{haah2023query} (see \cref{lem:phase-align-unitary}).

We now turn to the analysis of \cref{alg:base-approx-bd}.
We begin by proving a series of lemmas that will be important in the analysis.
Recall that $\paulisupport_{a,b} \subseteq \F_2^{2n}$ can be viewed as the set of phaseless Paulis $I^{\otimes n - a -b}\otimes \{I,Z\}^{b} \otimes \{I,X,Y,Z\}^{\otimes a}$. Let $\Pi_{\paulisupport_{a,b}}(U)$ denote the Pauli projection of $U$ onto the corresponding subspace. 
We will prove that $\Pi_{\paulisupport_{a,b}}(U)$ is exactly $(a,b)$-block diagonal and close to $U$. 

\begin{fact}\label{fact:block-diagonal-equiv}
    For a matrix $A$, the following are all equivalent conditions:
    \begin{enumerate}
        \item $A$ is $(a,b)$-block diagonal.
        \item $\supp(A) \subseteq \paulisupport_{a,b}$.
        \item $\Pi_{\paulisupport_{a,b}}(A) = A$.
    \end{enumerate}
\end{fact}
\begin{proof}
We start by showing that the first two conditions are equivalent.
Suppose $A$ is an $(a,b)$-block diagonal matrix. Then 
\[
A 
= I^{\otimes n - a - b} \otimes \left(\bigoplus_{y \in \{0,1\}^b} A_y \right) 
= I^{\otimes n - a - b} \otimes \sum_{y \in \{0,1\}^b} \ketbra{y}{y} \otimes A_y, 
\]
where each $A_y \in \C^{2^a \times 2^a}$.
It is easy to check that $b$-qubit operator $\ketbra{y}{y}$ is only supported on $\{I,Z\}^{\otimes b}$ for all $y \in \{0,1\}^b$. 
Each $A_y$ is trivially supported on $\{I,X,Y,Z\}^{a}$, because $\{I,X,Y,Z\}^{a}$ is a complete basis for operators on $a$-qubits. 
Hence, $\supp(A) \subseteq \paulisupport_{a,b}$. 
Conversely, suppose that $A$ is supported in $\paulisupport_{a,b}$.
Each operator in $\paulisupport_{a,b}$ is $(a,b)$-block diagonal, and so any linear combination of them is also $(a,b)$-block diagonal.

Now we show that conditions 2 and 3 are equivalent.
Let $A$ be an operator with $\supp(A) \subseteq \paulisupport_{a,b}$.
Then it is obvious that $\Pi_{\paulisupport_{a,b}}$ will have no effect on $A$, because it only affects Pauli coefficients outside of $\supp(A)$. 
Conversely, if $\Pi_{\paulisupport_{a,b}}(A)=A$, then all of the Pauli coefficients outside of $\paulisupport_{a,b}$ must be $0$, i.e., $\supp(A) \subseteq \paulisupport_{a,b}$.
\end{proof}

\begin{lemma}\label{lem:learn-projection}
    If $U$ is $(a,b,\eps)$-approximately block diagonal for any unitarily invariant norm $\norm{\cdot}$, then $\norm{U-\Pi_{\paulisupport_{a,b}}(U)} \leq 2\eps$.
\end{lemma}
\begin{proof}
Because $U$ is $(a,b,\eps)$-approximately block diagonal, there exists an $(a,b)$-block diagonal unitary $V$ such that $\norm{U - V}\leq \eps$.
Then
\begin{align}
    \norm{U - \Pi_{\paulisupport_{a,b}}(U)} 
    &\leq \norm{U - V} + \norm{V - \Pi_{\paulisupport_{a,b}}(U)} \\
    &= \norm{U - V} + \norm{\Pi_{\paulisupport_{a,b}}(V -U)} \\
    &= \norm{U - V} + \norm{\Ex_q\left[W_q(V -U)W_q\right]} \\
    &\leq \norm{U - V} + \Ex_q\left[\norm{W_q(V -U)W_q}\right] \\
    &= \norm{U - V} + \norm{U-V} \\
    &\leq 2\eps. 
\end{align}
The first line uses the triangle inequality. The second line uses the fact that $\Pi_{\paulisupport_{a,b}}(V) = V$ by \cref{fact:block-diagonal-equiv}. The third line follows from \cref{lemma:pauli-twirl}. The fourth line is the triangle inequality. The fifth line follows from the fact that the norm is unitarily invariant.
\end{proof}

Next, we show that if a unitary $U$ is close to some $(a,b)$-block diagonal matrix (not necessarily unitary), then $U$ is also close to a genuine $(a,b)$-block diagonal unitary. 
In other words, we can ‘round’ the approximate block-diagonal structure to an exact one without losing more than a constant factor in error. 
Moreover, this rounding can be done in polynomial time in the matrix dimension. 
One should view the non-unitary block diagonal matrix here as the approximation to $\Pi_{\paulisupport_{a,b}}(U)$ produced by the learning algorithm (i.e., before the SVD step in the second foreach loop).

\begin{lemma}\label{lem:round-to-unitary}
Let $U$ be a unitary and suppose there exists an $(a,b)$-block diagonal matrix $A$ (not necessarily unitary) such that $\norm{U - A} \leq \eps$ for some unitarily invariant norm $\norm{\cdot}$. Then $U$ is an $(a,b,2\eps)$-approximately block diagonal unitary.

Moreover, given a classical description of $A$, there is an algorithm that outputs an $(a,b)$-block diagonal unitary $V$ satisfying $\norm{U - V} \leq 2\eps$, and this algorithm runs in $O(2^{3a+b})$ time. 
\end{lemma}
\begin{proof}
We can write $A$ in block form as 
$A = I^{\otimes n-a-b} \otimes \left(\bigoplus_{y \in \{0, 1\}^b} A_y \right).$
For each block $A_y \in \C^{2^a \times 2^a}$, compute its SVD $A_y = L_y \Sigma_y R_y^\dagger$. 
Define $V \coloneqq I^{\otimes n-a-b} \otimes \left(\bigoplus_{y \in \{0, 1\}^b} L_y R_y^\dagger\right)$.
By construction, $V$ is $(a,b)$-block diagonal, and, in particular, it is the unitary obtained from the polar decomposition of $A$. 

It is a standard fact that, in any unitarily invariant norm, the unitary from the polar decomposition is the closest unitary to a given matrix. Thus,
\[
\norm{V - A} \leq \norm{U - A} \leq \eps,
\]
since $U$ is also unitary.
The triangle inequality gives
\[
\norm{U - V} \;\leq\; \norm{U - A} + \norm{A - V} \;\leq\; 2\eps,
\]
so $U$ is an $(a,b,2\eps)$-approximately block diagonal unitary.  

Finally, the runtime: computing an SVD (or equivalently, the polar decomposition) of a $2^a \times 2^a$ matrix takes $O(2^{3a})$ time. Since there are $2^b$ blocks, the total cost is $O(2^{3a+b})$.
\end{proof}

Together, \cref{lem:learn-projection,lem:round-to-unitary} show that the following two conditions are equivalent up to constant factors:
\begin{enumerate}
    \item $U$ is close to some $(a,b)$-block-diagonal unitary $V$, and 
    \item $U$ is close to its Pauli projection $\Pi_{\paulisupport_{a,b}}(U)$.
\end{enumerate}
For our analysis we will primarily analyze $U$ only through its Pauli projection.

The next lemma analyzes the cost of applying the Pauli projection $\Pi_{\paulisupport_{a,b}}(U)$. In particular, it bounds the number of queries to $U$ required in order to generate a sufficient number of copies of the projected state for use in the state tomography procedure.

\begin{lemma}\label{lem:post-selection-sample-complexity}
Let $U$ be an $(a,b,\eps)$-approximate block diagonal unitary,  and let $\ket{\psi}$ be any $n$-qubit state.
For any constant $c > 0$, there is a procedure that, given access to copies of $\ket{\psi}$ and using at most \[\frac{4}{1-2\eps} \cdot \frac{c \cdot 2^{a+b} + \log(1/\delta)}{\eps^2}\] queries to $U$, prepares $\frac{c \cdot 2^{a+b}}{\eps^2}$ copies of the state 
\[\frac{\Pi_{\paulisupport_{a,b}}(U)\ket{\psi}}{\norm{\Pi_{\paulisupport_{a,b}}(U)\ket{\psi}}_2}\] 
with probability at least $1-\exp(-2^{a+b})$.
\end{lemma}
\begin{proof}
By \cref{lemma:block-encoding-state-prep,lemma:lcu}, we successfully prepare a single copy of the desired state using a single query to $U$ with probability $\norm{\Pi_{\paulisupport_{a,b}}(U)\ket\psi}_2^2$.
By \cref{lem:learn-projection}, $\Pi_{\paulisupport_{a,b}}(U)$ is $2\eps$-close to a unitary in operator norm, so 
$\norm{\Pi_{\paulisupport_{a,b}}(U)\ket\psi}_2^2 \geq 1 - 2\eps$.

To obtain $d = \frac{c \cdot 2^{a+b}}{\eps^2}$ copies, we repeat the procedure over many independent trials. 
By \cref{lemma:bernoulli_bound_mult}, with $p \geq 1-2\eps$, $d = \frac{c \cdot 2^{a+b}}{\eps^2}$, and $\delta = \frac{\delta}{\exp(2^{a+b})}$, gives that
\[
M \leq \frac{4}{1 - 2\eps}  \cdot \frac{c \cdot 2^{a+b} + \log(1/\delta)}{\eps^2} 
\]
suffices.
\end{proof}

We also record the following basic fact, which says that rescaling the columns of a matrix can only perturb it by an amount proportional to the rescaling factor.
This is relevant when dealing with the $\frac{1}{\norm{\Pi_{\paulisupport_{a,b}}(U) \ket \psi}_2}$ factor in \cref{lemma:lcu}. 

\begin{fact}\label{lem:re-normalize-dist}
    Let $A$ be a matrix such that $\opnorm{A} \leq 1$.
    Then for arbitrary diagonal matrix $D \coloneqq \diag(d_1, \dots, d_{2^n})$, 
    $\opnorm{AD - A} \leq \max_{i} \abs{d_i - 1}$.
\end{fact}
\begin{proof}
    \[
        \opnorm{AD - A} 
        = \opnorm{A(D - I)}  
        \leq \opnorm{A}\opnorm{D - I} \leq \max_{i} \abs{d_i - 1}. \qedhere\]

\end{proof}

Finally, we also need two technical ingredients from \cite{haah2023query}. The first is a state tomography algorithm whose error lies in a Haar-random direction. The second is a post-processing routine that aligns the relative phases of the reconstructed unitary’s columns.

\begin{lemma}[{\cite[Proposition 2.2]{haah2023query}} and {\cite[Theorem 2, Section 4.3]{Guta_2020}}]\label{lem:state-tomo}
    There is a pure state tomography algorithm with the following behavior. Given access to copies of an $n$-qubit pure state $\ket \psi$, it sequentially and nonadaptively makes von Neumann measurements on $O\left(\frac{2^n + \log(1/\delta)}{\eps^2}\right)$ copies of $\ket \psi$.\footnote{As with \cite[Footnote 4]{haah2023query}, we can only make such measurements up to a certain machine precision depending on the underlying hardware. However, this only adds a time complexity that is dominated by other terms in \cref{alg:base-approx-bd}.} Then, after classically collating and processing the measurement outcomes in $O(8^n)$ time, it outputs (a classical description of) an estimate pure state 
    \[\ket{\wh{\psi}} = \phi \sqrt{1-\wh{\eps}^2} \ket \psi + \wh{\eps} \ket{w}\]
    such that:
    (1) $\phi$ is a complex phase;
    (2) the trace distance, $\wh{\eps}$, is at most $\eps$ except with probability at most $\frac{\delta}{\exp(5d)}$;
    (3) the vector $\ket w$ is distributed Haar-randomly among all states orthogonal to $\ket \psi$.
\end{lemma}

\begin{lemma}[Proof of {\cite[Proposition 2.3]{haah2023query}}]\label{lem:phase-align-unitary}
    Let $V, V' \in \C^{2^n \times 2^n}$ be classical descriptions of unitary matrices. Suppose there exist diagonal unitaries $\Phi_V, \Phi_{V'}$ and an operator $A\in\C^{2^n \times 2^n}$ such that $\opnorm{A \Phi_V - V} \leq \eps \leq \frac{1}{8}$ and $\opnorm{A H^{\otimes n} \Phi_{V'} - V'} \leq \eps \leq \frac{1}{8}$. 
    Then there is an algorithm that outputs a description of a diagonal unitary $D$ such that $\distphop(D, \Phi_V^\dagger) \leq 24\eps$ using $O(8^n)$ classical time. 
\end{lemma}

The proof of \cite[Proposition 2.3]{haah2023query} assumes that $A$ is unitary.
However, the argument does not use this, and in fact the proof goes through for arbitrary $A$.

With all the ingredients in place, we can now prove \cref{thm:approx-block-diag}, the main theorem of this section, which establishes the correctness of \cref{alg:base-approx-bd}.
We restate the theorem for convenience. 

\approxblockdiag*

\begin{proof}
    
    We assume $\eps \leq 1/C$ for some constant $C > 1$ to be fixed later. If instead $\eps > 1/C$, we may run \cref{alg:base-approx-bd} with $\eps = 1/C$, incurring only a constant-factor overhead of $C^2$, which is absorbed into the big-$O$ notation. In addition, we assume $a+b$ is a sufficiently large constant so that, for a universal constant $C'$ to be specified later, $\tfrac{a C'}{2^{a+b}} \leq \frac{1}{2}$. Note that if $a+b = O(1)$, then our desired query complexity can be achieved by na\"ively learning the entire $2^{a+b}\times 2^{a+b}$ block to Frobenius distance (as in \cite{chen2023testing}).
    This would still have $O(8^{a+b}/\eps^2) = O(1/\eps^2)$ query complexity as desired. 

    By \cref{fact:block-diagonal-equiv}, the Pauli projection $\Pi_{\paulisupport_{a,b}}(U)$ can be expressed in block-diagonal form as 
    \begin{equation}\label{eq:pauli-proj}
    \Pi_{\paulisupport_{a,b}}(U) = I^{\otimes n-a-b} \otimes \bigoplus_{y \in \{0, 1\}^b} B_y.
    \end{equation}
    Since $\Pi_{\paulisupport_{a,b}}(U)$ is obtained via a block-encoding (\cref{lemma:apply-pauli-projection}), we have $\opnorm{\Pi_{\paulisupport_{a,b}}(U)} \leq 1$.
    Furthermore, \cref{lem:learn-projection} gives that $\opnorm{\Pi_{\paulisupport_{a,b}}(U) - U} \leq 2\eps$. 
    Thus for every normalized state $\ket\psi$, $\norm{\Pi_{\paulisupport_{a,b}}(U)\ket\psi}_2 \geq 1-2\eps$. 

    Define $\alpha_z \coloneqq \norm{\Pi_{\paulisupport_{a,b}}\ket{0^{n-a-b}}\ket{+}^{\otimes b}\ket{z}}_2 \geq 1-2\eps$. 
    In \cref{alg-step:postselect}, we apply the postselection procedure of \cref{lemma:apply-pauli-projection} to prepare copies of the normalized state
     \begin{equation}\label{eq:state-tomo-input}
        \frac{\Pi_{\paulisupport_{a,b}}(U)\ket{0^{n-a-b}}\ket{+}^{\otimes b}\ket{z}}{\alpha_z} = \ket{0^{n-a-b}} \left(\frac{1}{\sqrt{2^b}} \sum_{y \in \{0, 1\}^b}\ket{y}\otimes \frac{B_y}{\alpha_z} \ket{z}\right) \eqqcolon \ket{0^{n-a-b}}\ket{\psi_z}
     \end{equation}
    which are then passed to the tomography step (\cref{alg-step:tomography}).
    By \cref{lem:post-selection-sample-complexity}, \[O\left(\tfrac{2^{a+b}}{\eps^2(1-2\eps)} \frac{2^a+b}{2^a}\right) = O\left(\frac{2^{a+b}}{\eps^2}(1+\frac{b}{2^a})\right)\] queries to $U$ suffice to produce the required number of copies, with failure probability at most $\exp(-5\cdot 2^{a+b})$. A union bound then guarantees that postselection succeeds for all $2^a$ columns except with probability at most $1/100$. Therefore, the overall query complexity of this stage (and the whole algorithm) is $O\left(\tfrac{2^{a+b}}{\eps^2}(2^a+b)\right)$. 

    Let us now focus on the output of \cref{lem:state-tomo} run on \cref{alg-step:tomography}. For each $z \in \{0,1\}^a$, the tomography algorithm produces a classical approximation $\ket{\wh{\psi_z}}$ of $\ket{\psi_z}$ defined in \cref{eq:state-tomo-input}.  
    Conditioned on the tomography succeeding, the algorithm returns a state of the form
    \begin{equation}\label{eq:state-tomo-output}
        \ket{\wh{\psi_z}} = \phi_z \sqrt{1-\eps_z^2} \left(\frac{1}{\sqrt{2^b}} \sum_{y \in \{0, 1\}^b}\ket{y}\otimes \frac{B_y}{\alpha_z} \ket{z}\right) + \eps_z \ket{w_z} 
        = \phi_z \sqrt{1 - \eps_z^2} \ket{\psi_z} + \eps_z \ket{w_z}
    \end{equation}
    where $\ket{w_z}$ is an $(a+b)$-qubit Haar-random state orthogonal to $\ket{\psi_z}$, $\eps_z \leq \eps \sqrt{\frac{2^a}{2^a+b}}$ is the accuracy of the tomography algorithm for the $z$th column, and $\phi_z$ is a random global phase. 
    By \cref{lem:state-tomo}, the tomography algorithm succeeds with probability at least $1-\exp(-5 \cdot 2^{a+b})$. By a union bound, all $2^a$ tomography algorithms succeed simultaneously except with probability at most $1/100$.
    In what follows, we condition on this success event. 

    Recall that our goal is to recover the matrices $B_y$ that appear in \cref{eq:pauli-proj}. The next step of the algorithm (\cref{alg-step:collate-1}) is to collate our state estimates (\cref{alg-step:tomography}) into block matrices $A_y$. For the analysis, express each $A_y$ in the form 
    \[
        A_y = B_y D \Phi E + W_y F, 
    \]
    where the matrices $D, \Phi, E, W_y, F$ are defined as follows.
    $D = \diag(\{\alpha_z^{-1}\})$ is the diagonal matrix of scaling factors $\alpha_z^{-1}$; $\Phi = \diag(\{\phi_z\})$ is the diagonal matrix of phases $\phi_z$; $E \coloneqq \diag(\{\sqrt{1-\eps_z^2}\}_z)$ and $F \coloneqq \diag(\{\eps_z\}_z)$ are the diagonal matrices of error terms.
    Lastly, for each $y \in \{0,1\}^b$, $W_y$ is the error matrix whose columns are given by 
    \[\sqrt{2^b} \cdot \left(\bra{y} \otimes I\right) \ket{w_z}.\]
    
    We will show that the operator norm of the error matrices $W_y$ is negligible using tools from random matrix theory. 
    To set this up, introduce uniformly random angles $\theta_z \in [0, 2\pi)$ and independent random variables random variables $\delta_1, \dots, \delta_{2^a}$, where each $\delta_z$ is distributed as $\abs{\braket{0 | \mathrm{Haar}}}^2$ for a Haar-random state $\ket{\mathrm{Haar}}$. In particular, each $\delta_z$ is a subexponential random variable with parameter $\norm{\delta_z}_{\psi_1} = \tfrac{1}{2^{a+b}}$. (Recall that $\norm{\cdot}_{\psi_1}$ is the subexponential norm, which is standard notation in probability theory.)
    Define
    \[
        \ket{h_z} \coloneqq \sqrt{\delta_z} e^{i \theta_z} \ket{\psi_z} + \sqrt{1-\delta_z} \ket{w_z} 
        = \sqrt{\delta_z} e^{i \theta_z} \frac{1}{\sqrt{2^b}} \sum_{y \in \{0, 1\}^b}\ket{y}\otimes \tfrac{B\prime_y}{\alpha_z} \ket{z} + \sqrt{1-\delta_z} \ket{w_z}.
    \]
    By construction, each $\ket{h_z}$ is an independent Haar-random state on $(a+b)$-qubits.
    Let $M \in \C^{2^{a+b} \times 2^a}$ be the matrix whose columns are the $\ket{h_z}$ and then let $M_y \coloneqq (\bra{y} \otimes I^{\otimes a}) M$. 
    Standard results in random matrix theory (\cite[Theorems 3.4.6 and 4.6.1]{Vershynin_2018}) show that for a fixed $y \in \{0, 1\}^b$, $\opnorm{M_y} \leq O\left(\frac{\sqrt{2^a+b}}{\sqrt{2^{a}}}\right)$ with probability at least $1-\exp(-5 \cdot (2^a+b))$.\footnote{We note that \cite{haah2023query} applies \cite{Vershynin_2018} in a similar manner. Note that each $\sqrt{2^{a+b}}\ket{w_z}$ is mean-zero, subgaussiank, and isotropic with $\norm{2^{a+b}\ket{w_z}}_{\psi_2} = 1$. Therefore, $\sqrt{2^a}(\bra{y} \otimes I^{\otimes a}) \ket{w_z}$ is mean-zero, subgaussian, and isotropic with parameter $\norm{\sqrt{2^a}(\bra{y} \otimes I^{\otimes a}) \ket{w_z}}_{\psi_2} = 1$.}
    By a union bound, all $M_y$ satisfy $\opnorm{M_y} \leq  O\left(\frac{\sqrt{2^a+b}}{\sqrt{2^{a}}}\right))$ with probability at least $1-\exp(-5 \cdot 2^a)$.
    Now decompose $\sqrt{2^b} M_y =  B_y D \Delta_0 + W_y \Delta_1$ where $\Delta_0 \coloneqq \diag(\{\sqrt{\delta_z} \cdot e^{-i \theta_z}\}_z)$ and $\Delta_1 \coloneqq \diag(\{\sqrt{1-\delta_z}\}_z)$.
    We therefore can bound $\opnorm{W_y}$ as follows. 
    \[
    \opnorm{W_y} \leq \left(\sqrt{2^b} \cdot \opnorm{M_y} + \opnorm{B_y D \cdot  \Delta_0}\right) \cdot \opnorm{\Delta_1^{-1}} \leq  \frac{O\left(\frac{\sqrt{2^a+b}}{\sqrt{2^{a}}}\right) + O\left(\max_z \sqrt{\delta_z}\right)}{\sqrt{1 - \max_z \delta_z}}.
    \]
    Because the $\delta_z$ are subexponential random variables with parameter $\norm{\delta_z}_{\psi_1} = \tfrac{1}{2^{a+b}}$, standard tail bounds \cite[Proposition 2.7.1, Lemma 2.7.6, Theorem 3.4.6]{Vershynin_2018} imply that, for any fixed $z$, $\delta_z \leq \tfrac{a C'}{2^{a+b}}$ except with failure probability at most $\frac{1}{100 \cdot 2^a}$ for a universal constant $C' > 0$. It is this constant $C'$ for which we need $\frac{aC^\prime}{2^{a+b}} \leq \frac{1}{2}$.
    By a union bound over all $2^a$ columns, this bound holds simultaneously for every $z$ except with probability at most $1/100$. Conditioned on this event, we have 
    \begin{equation}\label{eq:haar-term-small}
    \opnorm{W_y} \leq O\left(\frac{\sqrt{2^a+b}}{\sqrt{2^{a}}}\right)
    \end{equation}
    for all $y \in \{0,1\}^b$.

    Next, we have that for all $y \in \{0, 1\}^b$,
    \begin{align}\label{eq:per-block-bound}
        \opnorm{B_y\Phi - A_y} 
        &=\opnorm{B_y\Phi -B_yD\Phi} + \opnorm{B_y D \Phi - A_y} \nn 
        &\leq \opnorm{B_y - B_yD} + \opnorm{B_y D \Phi - A_y} \nn
        &= \opnorm{B_y - B_yD} + \opnorm{B_y D \Phi - B_y D \Phi E + W_y F} \nn
        &\leq \opnorm{B_y - B_yD} + \opnorm{B_y D} \opnorm{I - E} + \opnorm{W_y}\opnorm{F} \nn
        &\leq O(\eps). 
    \end{align}
    Recall that $\opnorm{I-E} = \max_{z \in \{0, 1\}^b} 1-\sqrt{1-\eps_z^2} \leq \eps_z \leq 1$ and $\opnorm{F} = \max_{z \in \{0, 1\}^b} \eps_z \leq  \eps \frac{2^a}{2^a+b}$.
    In the above, most steps follow from the triangle inequality and the submultiplicativity of the operator norm. In the second-to-last line, we use that, by \cref{lem:re-normalize-dist}, $\opnorm{B_y - B_y D } \leq \frac{1}{1-2\eps} - 1 = \frac{2\eps}{1-2\eps} \leq O(\eps)$ and we apply \cref{eq:haar-term-small}.
    Importantly, we have assumed that $\eps \leq 1/C$ for a universal constant $C > 0$. We take $C$ to be sufficiently large so that this bound is less than $\tfrac{1}{16}$. 
    Thus, when we round our output $A_y$ to the unitary $V_y$ in \cref{alg-step:collate-2}, the distance to $B_y \Phi$ will at most double to $\frac{1}{8}$ by \cref{lem:round-to-unitary}.

    We now correct for the relative phases in $\Phi$.
    By the group structure of $\paulisupport_{a,b}$, we have that
    \[
        \Pi_{\paulisupport_{a,b}}\left(U (I^{\otimes n-a} \otimes H^{\otimes a})\right) 
        = \Pi_{\paulisupport_{a,b}}(U) (I^{\otimes n-a} \otimes H^{\otimes a}) = I^{\otimes n-a-b} \otimes \left(\bigoplus_{y \in \{0, 1\}^b} B_y H^{\otimes a} \right)
    \]
    is an $(a,b)$-block diagonal matrix that is close to $U (I^{\otimes n-a} \otimes H^{\otimes a})$.
    By the same reasoning as the analysis above, \cref{alg-step:fix-phase1} recovers an $(a, b)$-block-diagonal unitary whose blocks $V'_y$ are at most $\frac{1}{8}$-far from $B_y H^{\otimes a} \Phi'$ for some other random phases $\Phi'$.%
    Hence, invoking \cref{lem:phase-align-unitary} on $V_0$ and $V_0^\prime$ in \cref{alg-step:fix-phase3} yields a diagonal unitary $D$ such that $\distphop(\Phi^\dagger, D) \leq O(\eps)$.

    Let $V = I^{\otimes n-a-b} \otimes \bigoplus_{y \in \{0, 1\}^b} V_y$ be the output of \cref{alg-step:merge-blocks}.
    By the triangle inequality and unitary invariance,
    \begin{align}\label{eq:fix-phases}
        \distphop(\Pi_{\paulisupport_{a,b}}(U), V (I^{\otimes n-a} \otimes D))
        &= \max_{y \in \{0, 1\}^b} \distphop(B_y, V_y D)\nn
        &\leq \max_{y \in \{0, 1\}^b} \distphop(B_y, V_y \Phi^\dagger) + \distphop(V_y \Phi^\dagger, V_y D)\nn
        &\leq \max_{y \in \{0, 1\}^b} \opnorm{B_y - V_y \Phi^\dagger} + \distphop(\Phi^\dagger,  D)\nn
        &\leq O(\eps). 
    \end{align}
    The last line follows from \cref{eq:per-block-bound} and the fact that $\distphop(\Phi^\dagger, D) \leq O(\eps)$.
    We note that it is essential to invoke \cref{lem:phase-align-unitary} on a $2^a \times 2^a$ block (rather than the entire unitary) to avoid an exponential dependence on $n$ in the final query and time complexity.

    We finally bound the distance between the output of the algorithm to $U$.
    Utilizing the triangle inequality once more:
    \begin{align*}
       \distphop(U, V (I^{\otimes n-a} \otimes D))
        &\leq \distphop\left(U, \Pi_{\paulisupport_{a,b}}(U) \right) + \distphop\left(\Pi_{\paulisupport_{a,b}}(U), V (I^{\otimes n-a} \otimes D)\right) \\
        &\leq \opnorm{U - \Pi_{\paulisupport_{a,b}}(U)} + O(\eps) \\
        &\leq O(\eps). 
    \end{align*}
    The second line follows from \cref{eq:fix-phases}, and the last line follows from \cref{lem:learn-projection}.

    The total failure probability is some small constant that we can suppress to an arbitrary $\delta$ by incurring a multiplicative $\log(1/\delta)$ factor in query complexity via standard amplification (see e.g., \cite[Proposition 2.4]{haah2023query}).
    The total runtime in $a+b$ is dominated by the complexity of running state tomography from \cref{lem:state-tomo}, which requires $O(8^{a+b})$ time, for $2^b \cdot \log(1/\delta)$ many repetitions.
    The $n$ and $\eps$ dependence is dominated by simply measuring the state for each tomography algorithm, which requires $O(n)$ per copy and $O\left(n \cdot \frac{2^{2a+b}}{\eps^2} (1+\frac{b}{2^a}) \log(1/\delta)\right)$ in total.
    The $\delta$ dependence is dominated by the $8^a 2^b \log^2(1/\delta)$ from \cite[Proposition 2.4]{haah2023query}.
    This gives a total runtime of
    \[
   O\!\left(
    n \tfrac{2^{a+b}}{\varepsilon^{2}} \left(2^a+b\right) \log \tfrac{1}{\delta}
    + 2^{3a+4b} \log \tfrac{1}{\delta}
    + 2^{3a+b} \log^{2} \tfrac{1}{\delta}
\right)
= \mathrm{poly}\!\bigl(n,\, 2^{2a+b},\, \varepsilon^{-1},\, \log \delta^{-1}\bigr).\qedhere
    \]
\end{proof}

\section{Reducing from \texorpdfstring{$k$}{\emph{k}}-Dimensionality to Block Diagonality}\label{sec:reduction-dimensionality}

In this section, we describe how to reduce the problem of learning a $k$-Pauli-dimensional unitary channel to the problem of learning an $(a,b,\eps)$-approximately block-diagonal unitary channel, for suitable integers $a,b$ with $2a+b \le k$.
Our reduction has two phases. In the first phase, we identify $\supp(U)$, the set of Pauli operators appearing in the expansion of the unknown unitary $U$. 
We give two algorithms for this task: one that requires only forward access to $U$, and another that additionally uses inverse access; the latter achieves a quadratic reduction in query complexity compared to the former.
In the second phase, we construct a Clifford circuit $C$ (using standard techniques) that maps the Pauli operators identified in the first phase into the canonical form $\paulisupport_{a,b}$.

Together, these two phases yield the desired reduction.
After the first phase, we obtain, with high probability, a subspace $S$ of Pauli operators that closely approximates $\supp(U)$.
Conjugating $U$ by the Clifford $C$ from the second phase maps $S$ into the canonical form $\paulisupport_{a,b}$, ensuring that $C U C^\dagger$ is $(a,b,\eps)$-approximately block-diagonal.
Consequently, the problem of learning $U$ reduces to that of learning an $(a,b,\eps)$-approximately block-diagonal unitary channel, which we solved in \cref{sec:learn-block-diag}.
We now present both phases in detail, which together form the first step of the complete algorithm described in \cref{sec:complete-algo}.

\subsection{Learning The Support}\label{ssec:learning-support}

The first step of our reduction is to identify the Pauli support of the unknown unitary $U$.
Recall that any $n$-qubit unitary can be expressed in the Pauli basis as
\[
U = \sum_{x \in \F_2^{2n}} \alpha_x W_x,
\]
where the coefficients form a probability distribution when squared in magnitude, since $\sum_x |\alpha_x|^2 = 1$.

A convenient way to access this distribution is via the Choi state of $U$, defined as
\[
\ket{\Phi_U} \coloneqq \left(I^{\otimes n} \otimes U\right) \frac{1}{\sqrt{2^n}} \sum_{x \in \F_2^n}\ket{x} \otimes \ket{x}. 
\]
The states $\{\ket{\Phi_{W_x}} : x \in \F_2^{2n}\}$ form the \emph{Bell basis}. It is a standard fact (see, e.g., \cite{montanaro2010quantum}) that measuring $\ket{\Phi_U}$ in the Bell basis yields outcome $x$ with probability $|\alpha_x|^2$. Thus, repeated Bell-basis measurements of the Choi state provide us with independent samples from the distribution supported on $\supp(U)$.

\subsubsection{Learning Without the Inverse}\label{subsec:learning-without-the-inverse}

Our strategy to learn $\supp(U)$ without the inverse is to simply collect enough Bell samples and infer a low-dimensional subspace that captures most of the probability mass. 
The following sampling lemma (proven in \cite{grewal2023efficient}) guarantees that a polynomial number of samples suffices.

\begin{lemma}[{\cite[Lemma 2.3]{grewal2023efficient}}]\label{lem:subspace-sample-complexity}
    Let $\calD$ be a distribution over $\F_2^d$, let $\eta, \delta \in (0, 1)$ and suppose \[
    m \geq 2\frac{d + \log(1/\delta)}{\eta}.
    \]
    Let $S \subseteq \F_2^d$ be the subspace spanned by $m$ independent samples drawn from $\calD$. Then \[
    \sum_{x \in S} \calD(x) \geq 1-\eta
    \]
    with probability at least $1-\delta$.
\end{lemma}

This result immediately gives the following corollary for quantum states.

\begin{corollary}\label{cor:sampling-without-inverse}
Let $\ket{\psi} \coloneqq \sum_{x \in \F_2^n} \beta_x \ket{x}$ be an $n$-qubit quantum state, and let $\eta, \delta \in (0,1)$.
Suppose there exists a subspace $D \subseteq \F_2^n$ of dimension $d$ such that
\[
    \sum_{x \in D} \abs{\beta_x}^2 = 1. 
\]
Then there is an algorithm that, with probability at least $1-\delta$, outputs a subspace $A \subseteq \F_2^n$ satisfying
\[
    \sum_{x \in A} \abs{\beta_x}^2 \geq 1-\eta
\]
using at most $m = O\left(\tfrac{d + \log(1/\delta)}{\eta}\right)$ copies of $\ket{\psi}$, no additional gate complexity, and \[O\left(mn \cdot \min\{m, n\}\right)\] classical post-processing time.
\end{corollary}
\begin{proof}
    By simply measuring $\ket \psi$ in the computational basis, we can apply \cref{lem:subspace-sample-complexity} to get enough samples, then perform Gaussian elimination to return a succinct set of generators for the span.
\end{proof}

Applying this to Bell samples from the Choi state $\ket{\Phi_U}$ gives us an efficient procedure for learning a subspace that contains nearly all of $\supp(U)$.

\subsubsection{Learning With the Inverse}\label{subsec:learning-with-the-inverse}

If we also have access to $U^\dagger$, then we can use amplitude amplification to obtain a more efficient support-learning algorithm.
The key ingredient is fixed-point amplitude amplification, which allows us to boost the probability of detecting computational-basis strings with non-negligible support.

\begin{lemma}[Fixed-point amplitude amplification {\cite{Yoder2014fixed}}]\label{lem:fixed-point-amplification}
    Let $\ket \psi$ be an $n$-qubit quantum state and $\Pi$ a diagonal projector in the computational basis such that $\abs{\braket{\psi | \Pi | \psi}}^2 \geq \eta$ for some known $\eta > 0$.
    Suppose we have unitaries $U, U^\dagger$ with $U\ket{0^n} = \ket \psi$.
    Then there is an algorithm that outputs a computational-basis state $\ket x$ satisfying $\braket{x|\Pi|x} = 1$ and $\braket{x|\psi}\neq 0$ with probability at least $1-\delta$, using  $O\left(\frac{\log(1/\delta)}{\sqrt{\eta}}\right)$ queries to $U, U^\dagger$ and $O\left(\frac{\log(1/\delta)}{\sqrt{\eta}}\cdot n\right)$ gate complexity.
\end{lemma}

This tool allows us to find basis states in the support of $\ket{\psi}$ more efficiently than by simple sampling. 
Combining it with the iterative spanning procedure from \cref{lem:subspace-sample-complexity} yields the following corollary, which improves the copy complexity by a quadratic factor in~$\eps$.

\begin{corollary}\label{cor:sampling-with-inverse}
Let $\ket{\psi} = \sum_{x \in \F_2^n} \beta_x \ket{x}$ be an $n$-qubit state, and let $\eta, \delta \in (0,1)$.
Suppose there exists a subspace $D \subseteq \F_2^n$ of dimension $d$ such that
\[
    \sum_{x \in D} \abs{\beta_x}^2 = 1.
\]
Then there is an algorithm that outputs a subspace $A \subseteq \F_2^n$ such that
\[
    \sum_{x \in A} \abs{\beta_x}^2 \geq 1-\eta 
\]
with probability at least $1-\delta$ using $O\left(\frac{d}{\sqrt{\eta}}\log(d/\delta)\right)$ queries to $U$ and $U^\dagger$, $O\left(\frac{dn}{\sqrt{\eta}} \log(d/\delta)\right)$ gate complexity. and $O\left(d^3 n\right)$ classical processing time.
\end{corollary}
\begin{proof}
We build up the target subspace iteratively. Let $A_0 = \{0\}$, and for $i=1,\dots,d$ repeat:
\begin{enumerate}
\item Define $\Pi_i$ to be the projector onto basis states outside of $A_{i-1}$.
\item Run \cref{lem:fixed-point-amplification} with parameters $\alpha=\eps$ and failure probability $\delta/d$, obtaining some $\ket{x}$ with $\braket{x|\Pi_i|x}=1$ and $|\braket{x|\psi}|>0$.
\item If such an $\ket{x}$ is obtained, update $A_i = \mathrm{span}(A_{i-1}, x)$ via Gaussian elimination; otherwise set $A_i = A_{i-1}$.
\end{enumerate}

    By a union bound, all $d$ calls to fixed-point amplification succeed with probability at least $1-\delta$, provided that at each iteration $i$ the current subspace $A_{i-1}$ still has weight $< 1-\eps$.
    If at some step the algorithm outputs an $\ket{x}$ with $\braket{x|\Pi_i|x}=0$, then by contrapositive this can only happen when $\sum_{x\in A_{i-1}}|\beta_x|^2 \geq 1-\eps$.
    Since $A_0 \subseteq A_1 \subseteq \cdots \subseteq A_d$, we conclude in this case that the final $A_d$ already satisfies $\sum_{x\in A_d}|\beta_x|^2 \geq 1-\eps$.

    On the other hand, if all $d$ iterations succeed in finding new basis states, then each $\ket{x}$ produced lies in $D$ (because the amplification subroutine always returns a string with nonzero overlap with $\ket{\psi}$).
    Thus after $d$ steps we have $A_d = D$.
    In either scenario, the output $A$ satisfies the desired guarantee.

    Finally, \cref{lem:fixed-point-amplification} uses $O\left(\tfrac{1}{\sqrt{\eps}}\log(d/\delta)\right)$ queries and $O\left(\tfrac{n}{\sqrt{\eps}}\log(d/\delta)\right)$ gates per iteration.
    Over $d$ iterations, this yields the claimed query and gate complexities.
    Each update of $A_i$ requires Gaussian elimination in $O(nd^2)$ time, so across all $d$ iterations the total classical cost is $O(d^3 n)$.
\end{proof}

\subsection{Mapping to a Block-Diagonal Unitary via Clifford Circuits}\label{ssec:map-to-block-diagonal}

Having identified an approximation to $\supp(U)$ in \cref{ssec:learning-support}, we now show how to map this support into the canonical form $\paulisupport_{a,b}$. 
Along the way, we also determine how good an approximation of $\supp(U)$ is needed for our reduction to go through.

A basic fact we rely on is that Clifford circuits normalize the Pauli group: for any Clifford $C$ and any $x \in \F_2^{2n}$, 
\[
    C W_x C^\dagger = \pm W_y
\]
for some $y \in \F_2^{2n}$. 
Thus, conjugating $U$ by a suitable Clifford circuit that transforms its Pauli support into $\paulisupport_{a,b}$ yields a unitary $C U C^\dagger$ that is block-diagonal. 
For a subspace $S \subseteq \F_2^{2n}$, we formally define
\[
    C S C^\dagger \coloneqq \{y \in \F_2^{2n} : \exists x \in S \text{ such that } C W_x C^\dagger = \pm W_y\}.
\]

Every Pauli subspace $S \subseteq \F_2^{2n}$ admits a canonical basis with respect to the symplectic inner product. 
Specifically, $S$ can always be generated by vectors of the form
\[
    \{x_1, z_1, \dots, x_a, z_a, z_{a+1}, \dots, z_{a+b}\},
\]
where the symplectic products satisfy $[x_i, x_j] = [z_i, z_j] = 0$ and $[x_i, z_j] = \delta_{ij}$ for all $i,j$. 
In this representation, $S$ naturally decomposes into a \emph{symplectic part} of dimension $2a$, spanned by the pairs $(x_i, z_i)$, and an \emph{isotropic part} of dimension $b$, spanned by $\{z_{a+1}, \dots, z_{a+b}\}$.
This $(a,b)$-decomposition directly determines the block structure: 
if $U$ has Pauli support contained in such an $S$, then after a suitable Clifford conjugation $U$ becomes $(a,b)$-block diagonal.

We can now state the main technical lemma of this section. 
It shows that if we learn $\supp(U)$ to sufficient accuracy, then we can efficiently construct a Clifford circuit $\widetilde{C}$ such that $\widetilde{C} U \widetilde{C}^\dagger$ is $(a,b,\eps)$-approximately block diagonal.

\begin{restatable}{lemma}{cliffordtoblock}\label{lem:clifford-to-block-junta}
    Let $U$ be a $k$-Pauli dimensional unitary, let $S \coloneqq \supp(U)$, and recall that $S$ admits an $(a,b)$-decomposition into a symplectic part of dimension $2a$ and an isotropic part of dimension $b$ (so that $2a+b = k$). 
    If one can learn a subspace $T \subseteq S$ satisfying
    \[
        \frac{1}{4^n}\sum_{x \in T} \abs{\tr(U W_x)}^2 \geq 1 - \frac{\eps^2}{2^{a+b+1}},
    \]
    then there exists a Clifford circuit $\widetilde{C}$ such that $\widetilde{C} U \widetilde{C}^\dagger$ is $(a^\prime, b^\prime, \eps)$-approximately block diagonal for some $a',b' \geq 0$ with $a^\prime + b^\prime \leq a+b$.
    Moreover, $\widetilde{C}$ can be constructed in time $O(nk^2)$.
\end{restatable}

To prove this lemma, we require two results about subspaces of symplectic vector spaces. 
The first is a semi-folklore decomposition procedure.
Specifically, given subspaces $T \subseteq S \subseteq \F_2^{2n}$, one can efficiently find a basis for $T$ split into symplectic and isotropic parts, and extend this to a basis for $S$ such that the basis of $T$ is contained in that of $S$. 
This result generalizes the Symplectic Gram-Schmidt procedure (see e.g., \cite[Lemma 2]{fattal2004entanglementstabilizerformalism}, \cite{wilde2009logical}, and \cite[Lemma 5.1]{grewal2023improved}) to work simultaneously on two subspaces (one a subspace of the other). 
An explicit algorithm for obtaining this decomposition is somewhat subtle, and, to the best of our knowledge, does not appear in the literature, so we provide it for completeness. 
The reader may safely skip these details of the proof in \cref{sec:deferred-proof} without missing any essential ideas of our main algorithm for learning $k$-dimensional unitary channels.

\begin{restatable}{lemma}{symplecticgramschmidt}\label{lem:symplectic-gram-schmidt}
    For subspaces $T \subseteq S \subseteq \F_2^{2n}$, there exist integers $a^\prime \leq a$ and $b^\prime \leq b$, together with an integer $\ell \leq b'$, such that
    \[T = \langle z_1, x_1, \dots, z_{a^\prime}, x_{a^\prime}, z_{a^\prime+1}, \dots, z_{a^\prime+b^\prime} \rangle\]
    and 
    \[S = \langle T, x_{a^\prime+1}, \dots, x_{a^\prime + \ell}, z_{a^\prime+b^\prime+1}, x_{a^\prime+b^\prime+1}, \dots, z_{a+b^\prime-\ell}, x_{a+b^\prime-\ell}, z_{a+b^\prime-\ell+1}, \dots, z_{a+b}\rangle \]
    where $[x_i, x_j] = [z_i, z_j] = 0$ and $[x_i, z_j] = \delta_{ij}$ for all $i,j$.
    Moreover, such a basis can be found in time $O(n (a+b)^2)$ given any generating sets for $T$ and $S$.
\end{restatable}

The second result is another semi-folklore result that can be attributed to the techniques of \cite{aaronson2004simulation}. This algorithm is a minor generalization of \cite[Section 4.2]{berg2021simple} and \cite[Lemma 3.2]{grewal2023efficient}.
\begin{lemma}\label{lem:standard-symplectic-basis}
    Given generators 
    \[\{z_1, x_1, \dots, z_{a^\prime}, x_{a^\prime}, z_{a^\prime+1}, \dots, z_{a^\prime+b^\prime}\}\]
    for a subspace $T \subseteq \F_2^{2n}$, and additional generators
    \[\{x_{a^\prime+1}, \dots, x_{a^\prime + \ell}, z_{a^\prime+b^\prime+1}, x_{a^\prime+b^\prime+1}, \dots, z_{a+b^\prime-\ell}, x_{a+b^\prime-\ell}, z_{a+b^\prime-\ell+1}, \dots, z_{a+b}\}\] 
    that extend this to a generating set for a larger subspace $S \supseteq T$, 
    where for all $i, j$, $[x_i, x_j] = [z_i, z_j] = 0$, and $[x_i, z_j] = \delta_{ij}$, there is a Clifford circuit $C$ such that $CTC^\dagger = \paulisupport_{a^\prime, b^\prime}$ and
    \[
        CSC^\dagger = 0^{n-a-b^\prime+\ell} \times \F_2^{a-a^\prime-\ell} \times 0^{b^\prime-\ell} \otimes \F_2^{a^\prime + \ell} \times 0^{n-a-b} \times \F_2^{a+b}. \footnote{When viewed as the set of phaseless Paulis, this becomes $I^{\otimes n-a-b} \otimes \{I, Z\}^{\otimes b-b^\prime+\ell} \otimes \{I, X, Y, Z\}^{\otimes a-a^\prime-\ell}  \otimes \{I, Z\}^{\otimes b^\prime-\ell} \otimes \{I, X, Y, Z \}^{\otimes a^\prime + \ell} $ such that it can be easily seen that $CSC^\dagger \subseteq \paulisupport_{a+b^\prime-\ell, b-b^\prime+\ell} \subseteq \paulisupport_{a+b, 0}$.}
    \] 
    Moreover, this Clifford circuit can be found in time $O(n(a+b))$.
\end{lemma}

Before proving the main lemma of this section, we record one final technical fact. It shows that if most of the Pauli weight of an $(a,b)$-block diagonal unitary is concentrated on a smaller subspace $\paulisupport_{a',b'}$, then the unitary is close (in operator norm) to its Pauli projection onto $\paulisupport_{a',b'}$.

\begin{lemma}\label{lem:pauli-support-to-op-bound}
Let $U$ be an $(a,b)$-block diagonal unitary, and suppose there exists $a',b' \in \N$ with $a^\prime + b^\prime \leq a + b$ such that 
\[
\frac{1}{4^n}\sum_{x \in \paulisupport_{a^\prime\!\!, b^\prime}} \abs{\tr(W_x U)}^2 \geq 1 - \frac{\eps^2}{2^{a+b}}
\]
and $\paulisupport_{a^\prime\!\!, b^\prime} \subseteq \supp(U)$.
Then $\opnorm{\Pi_{\paulisupport_{a^\prime, b^\prime}}(U) - U} \leq \eps$.
\end{lemma}
\begin{proof}
    Because $U$ is $(a, b)$-block diagonal, we can write $U = I^{\otimes n-a-b} \otimes U^\prime$ for some $a+b$ qubit unitary $U^\prime$ (recall that an $(a,b)$-block diagonal matrix is also $(a+b, 0)$-block diagonal).
    Observe that
    \[
        \frac{1}{2^n}\tr((I^{\otimes n-a-b} \otimes W_x) \cdot U) = \frac{1}{2^n}\tr(I^{\otimes n-a-b} \otimes (W_x\cdot U^\prime)) = \frac{1}{2^{a+b}} \tr(W_x U^\prime).
    \]
    
    Now let \[\Pi_{\paulisupport_{a^\prime, b^\prime}}(U) = I^{\otimes n-a-b} \otimes \Pi_{\paulisupport_{a^\prime, b^\prime}}(U^\prime).\footnote{Technically, these $\paulisupport_{a^\prime, b^\prime}$ are different, as the number of preceding identity matrices changes from $n-a-b$ to $(a+b) - (a^\prime + b^\prime)$.}\]
    By \cref{fact:-weyl-to-frob} and the fact that $\frac{1}{4^{a+b}}\sum_{x \in \supp(U^\prime)} \abs{\tr(W_x U^\prime)}^2 = 1$, 
    \begin{align*}
        \fnorm{\Pi_{\paulisupport_{a^\prime, b^\prime}}(U^\prime)-U^\prime}
        &= \sqrt{2^{a+b}}\cdot \sqrt{ \frac{1}{4^{a+b}}\sum_{x \in \supp(U^\prime)\setminus \paulisupport_{a^\prime, b^\prime}} \abs{\tr(W_x U^\prime)}^2}\\
        &= \sqrt{2^{a+b}}\cdot\sqrt{1 - \frac{1}{4^{a+b}}\sum_{x \in \paulisupport_{a^\prime, b^\prime}} \abs{\tr(W_x U^\prime)}^2}\\
        &= \sqrt{2^{a+b}}\cdot\sqrt{1 - \frac{1}{4^n}\sum_{x \in \paulisupport_{a^\prime, b^\prime}} \abs{\tr((I^{\otimes n-a-b} \otimes W_x) \cdot  U)}^2}\\
        &\leq \eps.
    \end{align*}
    Finally, note that
    \begin{align*}
        \opnorm{\Pi_{\paulisupport_{a^\prime, b^\prime}}(U) - U}
        &= \opnorm{I^{\otimes n-a-b}\otimes\left( \Pi_{\paulisupport_{a^\prime, b^\prime}}(U^\prime) - U^\prime\right)}\\
        &= \opnorm{\Pi_{\paulisupport_{a^\prime, b^\prime}}(U^\prime)-U^\prime} && (\text{\cref{fact:frob-repeat-identity}})\\
        &\leq \fnorm{\Pi_{\paulisupport_{a^\prime, b^\prime}}(U^\prime)-U^\prime} && (\text{\cref{fact:schatten-norm-conversion}})\\
        & \leq \eps.\qedhere
    \end{align*}
\end{proof}

By combining \cref{lem:symplectic-gram-schmidt,lem:standard-symplectic-basis,lem:pauli-support-to-op-bound}, we can prove \cref{lem:clifford-to-block-junta}.

\cliffordtoblock*
\begin{proof}
    Let us assume we know what $S$ is for a moment.
    \Cref{lem:symplectic-gram-schmidt} guarantees that $T$ and $S$ satisfy the requirements to run \cref{lem:standard-symplectic-basis} such that there exists a Clifford circuit $C$ where $CTC^\dagger = \paulisupport_{a^\prime, b^\prime}$ and $CSC^\dagger \subseteq \paulisupport_{a, b}$ for some $a^\prime +b^\prime \leq a+b$.
    Let $U_\text{BD} \coloneqq C U C^\dagger$.
    Using \cref{lem:pauli-support-to-op-bound,lem:round-to-unitary} we see that $U_\text{BD}$ is $(a^\prime, b^\prime, \eps)$-approximately block diagonal.

    Of course, this $C$ requires knowledge of $S$, which we don't actually have, but it still exists nevertheless, as does $U_{\text{BD}}$.
    Now let $\widetilde{C}$ be the unitary that \emph{just} takes $CTC^\dagger = \paulisupport_{a^\prime, b^\prime}$.\footnote{This is just a weaker statement than that of \cref{lem:standard-symplectic-basis}.}
    It can be found with just generators of $T$ in time $O(nk^2)$.
    Importantly, take (Clifford) unitary $C_2 = C \widetilde{C}^\dagger$, which normalizes
    \[
        C_2 \left(\paulisupport_{a^\prime, b^\prime}\right) C_2^\dagger = C_2 \left(\widetilde{C} T \widetilde{C}^\dagger \right) C_2^\dagger = C T C^\dagger = \paulisupport_{a^\prime, b^\prime}.
    \]
    Observe that this also implies
    \[
        C_2^\dagger \left(\paulisupport_{a^\prime, b^\prime}\right) C_2 = \paulisupport_{a^\prime, b^\prime}.
    \]
    
    As $U_\text{BD}$ is within $\eps$-close to some $(a^\prime, b^\prime)$-block diagonal unitary $V$ in operator distance, by unitary invariance
    \[
        \opnorm{U_\text{BD} - V} = \opnorm{C_2^\dagger U_{\text{BD}} C_2 - C_2^\dagger V C_2} \leq \eps.
    \]
    
    Finally, by \cref{fact:block-diagonal-equiv} note that $\supp(V) \subseteq \paulisupport_{a^\prime, b^\prime}$.
    Therefore $\supp\left(C_2^\dagger V C_2\right) \subseteq \paulisupport_{a^\prime, b^\prime}$, so $C_2^\dagger V C_2$ is \emph{also} $(a^\prime,b^\prime)$-block diagonal, by \cref{fact:block-diagonal-equiv} once again.
    As we established that $C_2^\dagger U_{\text{BD}} C_2$ is close to this $(a^\prime,b^\prime)$-block diagonal unitary $C_2^\dagger V C_2$, it follows that \[C_2^\dagger U_{\text{BD}} C_2 = \widetilde{C} C^\dagger \left(C U C^\dagger\right) C \widetilde{C}^\dagger = \widetilde{C} U \widetilde{C}^\dagger\] is $(a^\prime, b^\prime, \eps)$-approximately block-diagonal, completing the proof.
\end{proof}

\begin{remark}
The techniques of \cref{cor:sampling-without-inverse,cor:sampling-with-inverse,lem:clifford-to-block-junta} immediately yield property testers for $k$-Pauli dimensionality: an $O(k/\eps^2)$-query tester with only forward access to $U$, and an $O(k \log k / \eps)$-query tester when queries to the inverse $U^\dagger$ are also allowed, both in \nameref{def:frob-dist}.
Related results for $k$-juntas are known: Chen, Nadimpalli, and Yuen~\cite{chen2023testing} give a $\widetilde{O}(\sqrt{k}/\eps)$-query tester in the same distance measure, assuming inverse access.
We conjecture that a $\widetilde{O}(\sqrt{k}/\eps)$-query property tester for $k$-Pauli dimensionality should also be possible.
\end{remark}

\section{Learning \texorpdfstring{$k$-Pauli}{k-Pauli} Dimensionality}\label{sec:complete-algo}

In this section, we present our algorithms for learning $k$-Pauli-dimensional unitary channels and quantum $k$-juntas. 
We build on the previous two sections, which reduced the problem to learning approximately block-diagonal unitaries and gave algorithms for that task. This section is organized into three parts. First, we describe a base algorithm that learns $k$-Pauli-dimensional unitaries whose query complexity scales quadratically with the desired precision. Next, we apply the bootstrap technique of \cite{haah2023query} to upgrade this algorithm to achieve Heisenberg scaling. Finally, we show how our results yield a query-optimal algorithm for learning quantum juntas.

\subsection{Base Algorithm for \texorpdfstring{$k$-Pauli}{k-Pauli} Dimensionality}\label{subsec:base-algo}

We begin by giving our algorithm for learning $k$-Pauli-dimensional unitary channels whose query complexity scales quadratically with the desired precision.

\begin{theorem}\label{thm:pauli-dimension-base}
Let $a,b,n \in \N$ with $a+b \leq n$.
Let $U \in \C^{2^n \times 2^n}$ be a $k$-Pauli-dimensional unitary whose support is a $k$-dimensional Pauli subspace that admits a decomposition into a $2a$-dimensional symplectic part and a $b$-dimensional isotropic part, i.e., $\supp(C U) C^\dagger = \paulisupport_{a,b}$ for some Clifford circuit $C$.
There is a tomography algorithm that, given query access to $U \in \C^{2^n \times 2^n}$ and parameters $\delta, \eps > 0$, outputs an estimate $V$ satisfying $\distphop(U, V) \leq \eps$ with probability at least $1-\delta$. 
Moreover, the algorithm satisfies the following properties:
\begin{itemize}
    \item $\supp(V) \subseteq \supp(U)$.
    \item The algorithm makes at most $O\left(2^{a+b}\left(2^{a}+ b\right) \frac{\log (1/\delta)}{\eps^2}\right)$ queries to $U$ and only requires forward access, i.e., it does not require $U^\dagger$ or controlled-$U$).
    \item The algorithm runs in time $\poly\!\left(n,\, 2^{2a+b},\, \varepsilon^{-1},\, \log \delta^{-1}\right)$.
    \item The algorithm uses $\max\{2n - 2a -b, n\}$ additional qubits of space.
\end{itemize}
\end{theorem}
\begin{proof}
    To reduce from $k$-Pauli dimensional to $(a,b,\eps)$-approximately block diagonal, we will sample from the Choi state of $U$ to learn elements of its support. This requires $n$ extra ancilla qubits.
    Call the span of our samples (the number of which is to be determined momentarily) $A$.
    Using \cref{lem:clifford-to-block-junta} we need
    \[
        \frac{1}{4^n}\sum_{x \in A} \abs{\tr(W_x U)}^2 \geq 1 - \frac{\eps^2}{K \cdot 2^{a+b}}
    \]
    to find a Clifford circuit $C$ such that $C U C^\dagger$ is $(a,b,\eps/K^\prime)$-approximately block diagonal for sufficiently large constant $K^\prime$.
    By \cref{cor:sampling-with-inverse,cor:sampling-without-inverse} we only need $O\left( \frac{\sqrt{2^{a+b}}}{\eps} \cdot (2a+b) \cdot \log(\frac{2a+b}{\delta})\right)$ queries with the inverse and $O\left( \frac{2^{a+b}}{\eps^2} \cdot (2a+b) \cdot \log(1/\delta)\right)$ without the inverse.
    Finally, apply \cref{thm:approx-block-diag} to learn this $(a,b, \eps/K^\prime)$-approximately block diagonal unitary channel to $\eps$ in $\distphop(\cdot, \cdot)$.
    Each query to $C U C^\dagger$ only queries $U$ once so we end up using $O\left(2^{2a+b}/\eps^2\right)$ queries as desired.

    We need $n$ ancilla qubits to construct the Choi state of $U$, and $2n-(2a+b)$ qubits from \cref{thm:approx-block-diag}.
    The Clifford circuit from \cref{lem:clifford-to-block-junta} requires $O(n(a+b)^2)$ time.
    The time complexity of either \cref{cor:sampling-with-inverse} or \cref{cor:sampling-without-inverse} is $O\left( n \cdot \frac{\sqrt{2^{a+b}}}{\eps} \cdot (2a+b) \cdot \log(\frac{2a+b}{\delta})\right)$ with the inverse, and $O\left(\frac{n}{\eps^4} 4^{a+b} \cdot (2a+b)^2 \log^2(1/\delta)\right)$ without the inverse, respectively.
    Combined with the time complexity of \cref{alg:base-approx-bd}, we get a total complexity of $\poly\!\left(n,\, 2^{2a+b},\, \varepsilon^{-1},\, \log \delta^{-1}\right)$.
\end{proof}

Note that our algorithm admits two different time complexities: one when we only have access to $U$, and another when we additionally have $U^\dagger$ available.
The distinction arises solely in the support-learning phase, depending on whether we use the algorithm from \cref{subsec:learning-without-the-inverse} or the more efficient variant from \cref{subsec:learning-with-the-inverse}.
However, while \cref{subsec:learning-with-the-inverse} is also more query efficient, the query complexity of \cref{thm:pauli-dimension-base} is ultimately dominated by \cref{thm:approx-block-diag} so the only benefit is a better time complexity.\footnote{If, however, one could improve \cref{thm:approx-block-diag} to have query complexity $Q = o(2^{a+b}(2^a+b)$ then using \cref{cor:sampling-with-inverse} would lead to an improved query complexity of $Q$ as well. This would not hold when using \cref{cor:sampling-without-inverse}, as it would become the dominating subroutine.}

\subsection{Bootstrapping to Heisenberg Scaling}

To achieve the optimal $1/\eps$ Heisenberg scaling, we augment \cref{thm:pauli-dimension-base} using the bootstrapping technique of \cite{haah2023query}. 
The high-level idea is as follows. Suppose we first learn an approximation $V_1$ with $\distphop(U, V_1) \leq \eps$.
Then $UV_1^\dagger$ is itself $\eps$-close to the identity. 
Applying the base algorithm to $(UV_1^\dagger)^2$ yields a unitary $V_2$ such that $\distphop(V_2, (UV_1^\dagger)^2) \leq \eps$.
Taking a square root, $\sqrt{V_2}$ provides an approximation to $UV_1^\dagger$ that is accurate up to $\eps/2$, so that $\distphop(U, \sqrt{V_2} V_1) \leq \eps/2$.
Iterating this process with higher powers progressively reduces the error, ultimately yielding Heisenberg scaling.
The procedure is captured in the following lemma.

\begin{lemma}[{\cite[Theorem 3.3, Remark 3.4]{haah2023query}}]\label{lem:bootstrap}
Suppose we are given oracle access to an unknown unitary $U \in \C^{d \times d}$ belonging to a subgroup $G$ of unitaries that is closed under fractional powers (i.e., $U^{1/p} \in G$ for all $U \in G$ and $p \in \N$).
Assume further that there exists an algorithm $\calA$ that, given oracle access to $U \in G$, outputs a unitary $V \in G$ with with $\distphop(U,V) \leq \frac{1}{600}$ with probability at least $0.51$.
Then, for any error parameters $\eps, \delta \in (0,1)$, there is an algorithm that outputs a unitary $V^\prime \in G$ with the following guarantees:
\begin{itemize}
\item $\distphop(U, V^\prime) \leq \eps$ with probability at least $1-\delta$;
\item $\Ex\left[\distphop(U, V^\prime)^2\right] \leq (1+32\delta)\eps^2$;
\item $V^\prime \in G$.
\end{itemize}
Moreover, if $\calA$ uses $Q$ queries, then the new algorithm uses only $O\left(\frac{Q}{\eps}\log(1/\delta)\right)$ queries. 

Let $T$ denote the runtime of $\calA$ (to achieve constant error and constant success probability), $D$ the time to compute distances between elements in $G$, $P$ the time to compute a fractional power of a unitary in $G$, $M$ the time to multiply elements of $G$, and $S$ the time to synthesize a circuit for elements of $G$ given their classical descriptions.
Then the overall runtime of the new algorithm is

\[
    O\left(\left(\frac{S \cdot Q}{\eps} + \left(T+P+M + D\log (1/\delta)\right)\log(1/\eps)\right) \log(1/\delta)\right). \qedhere
\]
\end{lemma}

An immediate corollary of \cref{thm:pauli-dimension-base,lem:bootstrap} is an optimal algorithm for learning $k$-Pauli-dimensional unitary channels.
We emphasize that, for the bootstrapping to work, we crucially rely on the fact that the output $V$ of \cref{thm:pauli-dimension-base} satisfies $\supp(V) \subseteq \supp(U)$.
In other words, our learner is \emph{proper}: it always returns a unitary with support contained in that of $U$.

\begin{corollary}\label{cor:pauli-dimension-bootstrap-finegrain}
Let $a,b,n \in \N$ with $a+b \leq n$.
Let $U \in \C^{2^n \times 2^n}$ be a $k$-Pauli-dimensional unitary whose support is a $k$-dimensional Pauli subspace that admits a decomposition into a $2a$-dimensional symplectic part and a $b$-dimensional isotropic part, i.e., $\supp(CUC^\dagger) = \paulisupport_{a,b}$ for some Clifford circuit $C$.
There is a tomography algorithm that, given query access to $U \in \C^{2^n \times 2^n}$ and parameters $\delta, \eps > 0$, outputs an estimate $V$ satisfying $\distdiamond(U, V) \leq \eps$ with probability at least $1-\delta$. 
Moreover, the algorithm satisfies the following properties:
\begin{itemize}
    \item $\supp(V) \subseteq \supp(U)$.
    \item The algorithm makes at most $O\left(2^{a+b}(2^a+b) \frac{\log (1/\delta)}{\eps}\right)$ to $U$ and only requires forward access, i.e., it does not require $U^\dagger$ or controlled-$U$).
    \item The algorithm runs in time $\poly\!\left(n,\, 2^{2a+b},\, \varepsilon^{-1},\, \log \delta^{-1}\right)$.
    \item The algorithm uses between $n$ and $2n-1$ additional qubits of space.
\end{itemize}
\end{corollary}
\begin{proof}
    Let $A = \langle \supp(U)\rangle$ be the $k$-dimensional subspace that contains $\supp(U)$ and let $G$ be the unitary subgroup involving unitaries whose support lies within $A$.
    Because \cref{thm:pauli-dimension-base} always outputs an element of $G$, and because $G$ is closed under fractional powers, we can combine \cref{thm:pauli-dimension-base} and \cref{lem:bootstrap} to get Heisenberg scaling.
    Finally, use \cref{fact:op-to-diamond} to convert the distance to diamond distance.

    Using our reduction of \cref{lem:clifford-to-block-junta} to the block-diagonal case, we can compute the distance between elements of $G$, arbitrary fractional powers, and multiply elements all in time $O(8^a \cdot 2^b + n(2a+b)^2)$.\footnote{This technically requires full knowledge of $G$, rather than an approximation of $G$ like we will get from \cref{lem:clifford-to-block-junta,thm:pauli-dimension-base}. However, we only need to compute these values relative to the outputs of \cref{thm:pauli-dimension-base}, all of which we always know the exact support of.}
    Constructing an arbitrary element of $G$ can be done using $O(4^a \cdot 2^b \cdot \poly\log(2^{a+b}/\eps))$ many additional gates using the Solovay-Kitaev theorem.
    Including the runtime of \cref{thm:pauli-dimension-base}, we find the total runtime to be $\poly\!\left(n,\, 2^{2a+b},\, \varepsilon^{-1},\, \log \delta^{-1}\right)$.
\end{proof}

When $b = O(2^a)$, the algorithm is provably query-optimal (up to a constant).
An important instance of this is our algorithm for quantum $k$-juntas (\cref{subsec:juntas}), where $a=k$ and $b=0$.
We conjecture that the version without inverse access is query optimal for all parameters. Establishing this would likely require techniques along the lines of \cite{tang2025amplitudeamplificationestimationrequire}.

\begin{remark}[Time complexities of our algorithm]\label{remark:time-complexity}
    While calculating the exact runtime of \cref{cor:pauli-dimension-bootstrap-finegrain} is tricky, a rough upper bound on the complexity is:
    \[
        \widetilde{O}\left(\left(2^{3a+b} \left(\frac{2^{a+b}}{\eps} + 8^b + \log(1/\delta)\right) + n\left(4^{a+b} + \log(1/\delta)\right)\right) \log(1/\delta)\right)
    \]
    using only forward access to $U$.
    If one were to use \cref{cor:sampling-with-inverse}, then the resulting runtime would be
    \[
        \widetilde{O}\left(\left(2^{3a+b}\left(\frac{2^{a+b}}{\eps} + 8^b + \log(1/\delta)\right) + n\left(2^{2a+b} + \log(1/\delta)\right)\right) \log(1/\delta)\right)
    \]
    using queries to $U^\dagger$ as well.
\end{remark}

Finally, let us state the performance of our algorithm solely in terms of $k$ (as $a$ and $b$ are generally unknown quantities). Because $2a+b = k$ and $a+b \leq k$, this follows easily from \cref{cor:pauli-dimension-bootstrap-finegrain}.

\begin{corollary}\label{cor:pauli-dimension-bootstrap}
Let $U \in \C^{2^n \times 2^n}$ be a $k$-Pauli dimensional unitary.
There is a tomography algorithm that, given query access to an $k$-Pauli dimensional unitary $U \in \C^{2^n \times 2^n}$ as well as parameters $\delta, \eps > 0$, outputs an estimate $V$ satisfying $\distdiamond(U, V) \leq \eps$ with probability at least $1-\delta$. 
Moreover, the algorithm satisfies the following properties:
\begin{itemize}
    \item $\supp(V) \subseteq \supp(U)$.
    \item the algorithm makes at most $O\left(2^k \cdot k\frac{\log (1/\delta)}{\eps}\right)$ queries to $U$ and only requires forward access, i.e., it does not require $U^\dagger$ or controlled-$U$).
    \item The algorithm runs in time
        \[
        \widetilde{O}\left(4^k\left(n + \frac{1}{\eps}\right)\log(1/\delta) + \left(n + \left(2 \sqrt{2}\right)^k \right) \log^2(1/\delta)\right)
    \] when only given forward access to $U$. 
    \item The algorithm runs in time 
    \[
        \widetilde{O}\left(2^k\left(n + \frac{2^k}{\eps}\right)\log(1/\delta) + \left(n + \left(2 \sqrt{2}\right)^k \right) \log^2(1/\delta)\right)
    \] when given access to both $U$ and $U^\dagger$. 

    \item The algorithm uses between $n$ and $2n-1$ additional qubits of space.
\end{itemize}
\end{corollary}

\section{Learning \texorpdfstring{$s$}{s}-Pauli Sparsity}
\label{sec:sparsity}

In this section, we present our algorithm for learning unitaries with sparse Pauli support. 
At a high level, we reduce the problem to quantum state learning. 
In particular, we show that obtaining a $\poly(s)$-time \emph{improper} learning algorithm reduces to the following state tomography task: given copies of a state $\ket{\psi}$ supported on $s$ computational basis states, output an approximate classical description of $\ket{\psi}$ in $\poly(s)$ time.
We then provide an optimal algorithm for this state tomography task.

The resulting algorithm outputs an estimate $\hat{A}$ that is not necessarily unitary. 
To obtain a proper learner, one must round $\hat{A}$ to a nearby unitary that remains sparse and close to the unknown unitary $U$. 
It is unclear how to perform this rounding efficiently in general, and we leave this as an interesting open problem. 
We do, however, identify several natural settings in which efficient rounding is possible; these are discussed at the end of this section.
All circuit classes considered in this work fall into these settings, and hence all of our resulting learning algorithms run in polynomial time.
We note that the classical rounding step can be avoided altogether if one only needs to apply a quantum channel that close to the unknown unitary; see \cref{remark:quantum-polar} for details. 

\begin{definition}[Pauli sparsity]
    A matrix $A$ is $s$-Pauli sparse if $\abs{\supp(A)} \leq s$.
\end{definition}

The following lemma establishes that we can efficiently learn the support of the unknown sparse unitary.

\begin{lemma}\label{lem:sparse-sample-bound}
    Given an unknown $s$-Pauli sparse $n$-qubit unitary $U = \sum_{x \in \supp(U)} \alpha_x W_x$, using $2\frac{s+\log(1/\delta)}{\eps^2}$ queries, $O\left(\frac{n}{\eps^2}(s + \log(1/\delta))\right)$ time, and $n$-ancilla qubits, one can learn a set $S \subseteq \supp(U)$ such that \[\sum_{x \not \in S} \abs{\alpha_x}^2 \leq \eps^2.\]
\end{lemma}
\begin{proof}
    We will use Bell sampling on $U$ to learn $\supp(U)$. Let $Y_1, \dots, Y_m$ be (dependent) the Bernoulli random variable that is $1$ if we learn a new element of $\supp(U)$ \emph{or} if we have already learned enough element of $\supp(U)$ to satisfy $\sum_{x \not \in S} \abs{\alpha_x}^2 \leq \eps^2$.
    Otherwise let $Y_i$ be zero.
    Observe that if $\sum_{i=1}^m Y_i \geq s$ then we have succeeded, as we have either learned all of $\supp(U)$ or $\sum_{x \not \in S} \abs{\alpha_x}^2 \leq \eps^2$.
    
    We can see that $\Pr\left[Y_i = 1\right] \geq \eps^2$, independent of the other random variables.
    So we can define i.i.d Bernoulli random variables $Z_i$ such that $Y_i = Z_i + E_i$ for $E_i = 1$ iff $\sum_{x \not \in S} \abs{\alpha_x}^2 \leq \eps^2$ \emph{and} $Z_i = 0$, and $E_i = 0$ otherwise.
    It follows by \cref{lemma:bernoulli_bound_mult} that $\Pr\left[\sum_{i=1}^m Y_i \leq s\right] \leq \Pr\left[\sum_{i=1}^m Z_i \leq s\right] \leq \delta$ if $m = \frac{2}{\eps^2}(s + \log(1/\delta))$.
\end{proof}

Similar to \cref{subsec:learning-with-the-inverse}, with the inverse one can use \cref{lem:fixed-point-amplification} to get Heisenberg scaling using $O(\frac{s}{\eps}\log(s/\delta))$ query complexity.

Next, we give an optimal state tomography algorithm for states supported on a small number of computational basis states. 

\begin{lemma}[Copy-optimal tomography of sparse quantum states]\label{lem:sparse-tomo-state}
Let $\ket{\psi}$ be an $n$-qubit quantum state supported on at most $s$ computational basis states.
Given copies of $\ket{\psi}$, there is an algorithm that outputs a classical description of a state $\ket{\hat{\psi}}$ such that, with probability at least $1-\delta$,
\begin{itemize}
    \item $\ket{\hat{\psi}}$ is $\eps$-close to $\ket{\psi}$ in trace distance, for $\eps \in (0,1]$;
    \item $\ket{\hat{\psi}}$ is supported on a subset of the support of $\ket{\psi}$;
    \item the algorithm uses at most $O\!\left(\frac{s + \log(1/\delta)}{\eps^2}\right)$ copies of $\ket{\psi}$;
    \item the algorithm runs in time $O\!\left(ns\,\frac{s + \log(1/\delta)}{\eps^2} + s^3\right)$.
\end{itemize}
Moreover, $\Omega(s/\eps^2)$ are necessary for this task.
\end{lemma}
\begin{proof}
    Let $\supp(\psi) \subseteq \F_2^n$ be the computational basis state that $\ket \psi$ is supported on and define $\ket\psi \coloneqq \sum_{x \in \supp(\psi)} \beta_x \ket{x}$.
    We can start by measuring in the computational basis to learn $\supp(S)$.
    Using \cref{lem:sparse-sample-bound}, we can learn a subset $S \subseteq \supp(\psi)$  such that $\sum_{x \in S} \beta_x^2 \geq 1-\eps^2/4$ using $\frac{s + \log(1/\delta)}{\eps^2}$ samples to $\ket \psi$ and $O\left(n \frac{s + \log(1/\delta)}{\eps^2}\right)$ time.

    Now let $\ket \phi \coloneqq \frac{1}{\sqrt{\sum_{x \in S} \beta_x^2}}\sum_{x \in S}\beta_x \ket{x}$ be the pure state that results in post-selecting on getting an outcome in $S$ when measuring in the computational basis.
    Such a post-selection takes $O(sn)$ time per sample to go through a list of size $s$ of $n$-bit strings to check for inclusion.
    See that this is just a state on an $s$-dimensional space.
    Using \cref{lem:state-tomo}, we can get an $\eps/2$-distance estimate of $\ket \phi$ using at most $O\left(\frac{s + \log(2/\delta)}{\eps^2}\right)$ samples and $O\left(s\frac{s + \log(1/\delta)}{\eps^2} + s^3\right)$ time with probability at least $1-\delta/2$.
    By \cref{lemma:bernoulli_bound_mult} and a union bound, we only need \[\frac{2}{1-\eps^2/4}\left(\frac{s + \log(2/\delta)}{\eps^2} + \log(2/\delta)\right) \leq O\left(\frac{s + \log(1/\delta)}{\eps^2}\right)\] post-selections of $\ket \psi$ to get the needed samples of $\ket \phi$.

    Finally, note that the trace distance between $\ket \phi$ and $\ket \psi$ goes as 
    \begin{align*}
        \sqrt{1 - \abs{\braket{\phi | \psi}}^2} &\leq \sqrt{1 - \left|\left(\frac{1}{\sqrt{\sum_{x \in S} \beta_x^2}} \sum_{x \in S} \beta_x^* \bra{x}\right)\left(\sum_{y \in \supp(\psi)} \beta_y \ket{y}\right)\right|^2}\\
        &\leq \sqrt{1- \sqrt{\sum_{x \in S} \abs{\beta_x}^2}}\\
        &\leq \sqrt{1-\sqrt{1-\eps^2/4}}\\
        &\leq \eps/2.
    \end{align*}
    So by the triangle inequality, a $\eps/2$ estimate of $\ket{\phi}$ results in a $\eps$-error estimate of $\ket \psi$ in trace distance.

    For the lower bound, it is known that $\Omega(d/\eps^2)$ copies are necessary to learn a $d$-dimensional pure state to $\eps$ accuracy in trace distance~\cite{scharnhorst2025optimallowerboundsquantum}.
    Thus, the lower bound $\Omega(s/\eps^2)$ follows by observing that if an $o(s/\eps^2)$ algorithm exists, it would contradict this lower bound. 
\end{proof}

We can now present our algorithm for learning unitaries with sparse Pauli support. The algorithm goes by reducing to the task solved in \cref{lem:sparse-tomo-state}.

\begin{theorem}[Tomography of Pauli-sparse unitary channels]
\label{thm:learning-sparsity}
    Let $s \in \N$ with $s \leq n$.
    Let $U \in \C^{2^n \times 2^n}$ be a $s$-Pauli-sparse unitary.
    There is a tomography algorithm that, given query access to $U \in \C^{2^n \times 2^n}$ and parameters $\delta, \eps > 0$, outputs a classical description of a matrix $V$ satisfying $\distphop(U, V) \leq \eps$ with probability at least $1-\delta$. 
    Moreover, the algorithm satisfies the following properties:
\begin{itemize}
    \item $\supp(V) \subseteq \supp(U)$.
    \item The algorithm makes at most $O\left(\frac{s}{\eps^2}(s+ \log(1/\delta))\right)$ queries to $U$ and only requires forward access, i.e., it does not require $U^\dagger$ or controlled-$U$).
    \item The algorithm runs in time $O\left(ns^2 \frac{s+\log(1/\delta)}{\eps^2} + s^3\right)$.
    \item The algorithm uses $n$ additional qubits of space.
\end{itemize}
\end{theorem}
\begin{proof}
    Observe that the Choi state of $U$ can be defined as $\ket{\Phi_U} \coloneqq \sum_{x \in \supp(U)} \alpha_x \ket{\Phi_{W_x}}$ when written in the Bell basis.
    If we can learn the Choi state of $U$ using only the Bell states that $\ket{\Phi_U}$ is composed of to $\frac{\eps}{\sqrt{s}}$ error in trace distance, then we have learned a $2n$-qubit state $\ket{\hat{\psi}} \coloneqq \sum_{x \in \supp(S)} \hat{\alpha_x} \ket{\Phi_{W_x}}$ such that (up to global phase) $\sqrt{\sum_{x \in \supp(U)} \abs{\alpha_x - \hat{\alpha_x}}^2} \leq \frac{\eps}{2\sqrt{s}}$.\footnote{Note that $\ket{\hat{\psi}}$ is not necessarily a valid Choi state, which is why we do not refer to it as something like $\ket{\Phi_{\hat{U}}}$.}
    By Cauchy-Schwarz, this says that $\sum_{x \in \supp(U)} \abs{\alpha_x - \hat{\alpha_x}} \leq \eps/2$.
    Therefore, if we can define the matrix $A \coloneqq \sum_{x \in \supp(U)} \hat{\alpha_x} W_x$, then by the triangle inequality it is $\eps/2$-close to $U$ in \nameref{def:phase-op-dist}.
    It therefore suffices to perform sufficiently accurate state tomography on $\ket{\Phi_U}$.

    Moreover, $\ket{\Phi_U}$ can be viewed as $s$-sparse in the Bell basis.
    This means that tomography of $\ket{\Phi_U}$ reduces to that of tomography of an $s$-sparse $2n$-qubit quantum state in the \emph{computational basis}.
    Using \cref{lem:sparse-tomo-state}, we can get such an estimate of $\ket{\Phi_U}$ using $O\left(\frac{s}{\eps^2}(s+ \log(1/\delta))\right)$ samples.
\end{proof}

One can replace \cref{lem:sparse-tomo-state} with alternative algorithms for sparse state tomography. 
For example, \cite[Theorem 26]{Apeldoorn2023tomography} yields $\widetilde{O}(s^{1.5}/\eps)$ scaling, at the cost of requiring inverse queries to $U$.
As another example, suppose one is given a set $S$ such that 
\[
\sqrt{\sum_{x \in \supp(U)} \abs{\alpha_x - \hat{\alpha}_x}^2} \;\le\; \frac{\eps}{2\sqrt{s}},
\]
as in the proof of \cref{thm:learning-sparsity}. 
Then we believe that the algorithm of \cite{chen2025inversefreequantumstateestimation} can be used to obtain an $O(s^2/\eps)$-query algorithm using only forward queries.

As discussed, \cref{thm:learning-sparsity} yields an \emph{improper} learner, as it outputs a matrix that is not necessarily unitary. 
Information-theoretically, this output suffices to recover a nearby unitary. For example, one could compute the nearest unitary via the polar decomposition, but this requires $\exp(n)$ time. 
Ideally, one would instead have a rounding procedure that outputs a nearby unitary that remains $s$-sparse and can be computed efficiently. 
At present, it is unclear how to perform such rounding efficiently (or whether it is possible in general), and we leave this as an open problem. 
Below, we identify three natural settings in which this rounding can be carried out efficiently while preserving sparsity.

\begin{fact}\label{fact:very-sparse-rounding}
    Any matrix $A$ that can be expressed as a sum over $s < n$ Weyl operators can be rounded to its closest unitary operator (in any unitarily invariant distance) in time $O(\exp(s))$ while increasing the Pauli sparsity to at most $2^s$.
\end{fact}
\begin{proof}
The Pauli sparsity being $s$ also implies that the Pauli dimension is at most $s$.
This lets one compute the polar decomposition in time $O(\exp(s))$.\footnote{We note that for Fourier sparsity $s$, the Fourier dimensionality is bounded above by $O(\sqrt{s}\log s)$~\cite{sanyal2019fourier}. It is an open problem if a similar non-trivial bound exists for Pauli sparsity and Pauli dimensionality.}
\end{proof}

\begin{fact}\label{fact:commute-rounding}
    Let $s \in \mathbb{N}$ and let $A$ be a matrix that can be expressed as a sum over at most $s$ mutually commuting Weyl operators and let $U$ be the closest Hermitian unitary operator to $A$ in operator distance with Pauli sparsity $s$ and $\supp(A) \subseteq \supp(U)$.
    If $\opnorm{U-A} < 1/2$ then $U$ can be identified from a classical description of $A$ in time $O(ns^2)$.
\end{fact}
\begin{proof}
    Let $A \coloneqq \sum_{x \in \supp(A)} \alpha_x W_x$ represent the sum over mutually commuting Weyl operators.
    Because operator distance is unitarily invariant, WLOG we can assume that the mutually commuting Weyl operators in $\supp(A)$ lie in $0^n \times \F_2^n$ (i.e., Pauli $\{I, Z\}^{\otimes n}$ operators) such that $\supp(A) \subseteq 0^n \times \F_2^n$.
    This is because we can find a Clifford circuit that maps $\supp(A)$ to $0^n \times \F_2^n$ in time $O(ns^2)$ using \cref{lem:symplectic-gram-schmidt}, which will be the dominating time complexity.
    The result is that  $A$ is simply a diagonal matrix and $U$ only has real-valued Weyl coefficients.
    The resulting closest Hermitian unitary $U$ is then a diagonal matrix with a Boolean function $f : \F_2^n \rightarrow \{\pm 1\}$ along the diagonal.

    To make the Weyl coefficients easily computable, it suffices to the Weyl coefficients $A$ be real numbers, and then further round these real numberrs to multiples of $\frac{1}{4s}$, to generate matrix $A^\prime \in \mathbb{R}^{2^n \times 2^n}$.
    This only moves the real part of each Weyl coefficient by at most $\frac{1}{2s}$, such that by the triangle inequality the operator distance of this $A^\prime$ to $U$ is at most $\frac{1}{2} + \eps < 1$.
    By the fact that $\opnorm{U-A}^\prime \leq \eps < 1$, the real part of the Weyl coefficients of $A$ must have the same sign as that of $U$. 
    It therefore suffices to compute the sign of the real part of the diagonal of $A$ to get $f$.
    
    As this function is a weighted sum over $s$ Parity functions, we can use a threshold gate to see if this sum is positive or negative to implement this Boolean function.
    Since the weights are multiples of $\frac{1}{4s}$ such that the $\ell_1$ norm of these coefficients is bounded by $O(\sqrt{s})$, this reduces to implementing a Majority gate on $O(s^{1.5})$-bits (each bit is the output of a Parity function).
    As the parity function and majority function on $k$-bits have circuit complexity $O(k)$ and circuit depth $O(\log k)$, the total circuit implementing $U$ is implementable by a $O(\log (sn))$-depth classical (resp. quantum) circuits with circuit complexity $O(s^{1.5} + n)$ and is therefore efficient.
\end{proof}

\begin{fact}\label{fact:majorana-rounding}
    Any Hermitian matrix $A$ that can be expressed as a sum over mutually anti-commuting $s$ Weyl operators can be rounded, in time $O(ns^2)$, to it's closest Hermitian unitary operator $U$ (in Frobenius distance) such that $\supp(U) = \supp(A)$.
\end{fact}
\begin{proof}
    Let $S$ be a set of mutually non-commuting Weyl operators such that $\abs{S} = s$ and let $A \coloneqq \sum_{x \in S} \alpha_x W_x$ be a Hermitian matrix, meaning that $\alpha_x$ are real-valued.
    Then we can see that
    \begin{align*}
        A^2 &= \sum_{x, y \in S} \alpha_x \alpha_y W_x W_y\\
        &= I \cdot (\sum_{x \in S} \alpha_x^2) +\sum_{x < y \in S } (\alpha_x \alpha_y - \alpha_x\alpha_y) W_x W_y\\
        &= I \cdot (\sum_{x \in S} \alpha_x^2).
    \end{align*}
    It follows that $\frac{A}{\sqrt{\sum_{x \in S} \alpha_x^2}}$ must be a Hermitian unitary operator.

    To show that this is the closest Hermitian unitary operator with the same Pauli support, note that the Frobenius distance between Hermitian operator $A$ and Hermitian unitary $U$ with $\supp(A), \supp(U) \subseteq S$ such that $A = \sum_{x \in S} \alpha_x W_x$ and $Y = \sum_{y \in S} \beta_y W_y$ can be defined as
    \[
        \norm{U-V}_F^2 = \sum_{x \in S} (\alpha_x - \beta_x)^2 = \sum_{x \in S} \alpha_x^2 + \beta_x^2 - 2\alpha_x \beta_x = \left(\sum_{x \in S} \alpha_x^2 + 1\right) -2 \sum_{x \in S} \alpha_x \beta_x.
    \]
    Therefore, if $A$ is fixed as the Hermitian matrix we want to be close to, it follows that we simply want to maximize $\sum_{x \in S} \alpha_x \beta_x$ subject to $\sum_{x \in S} \beta_x^2 = 1$, as $\sum_{x \in S} \alpha_x^2$ is fixed.
    This is easily done by setting $\beta_x = \frac{\alpha_x}{\sqrt{\sum_{x \in S} \alpha_x^2}}$.

    To actually efficiently synthesize a circuit for this, we can use \cref{lem:symplectic-gram-schmidt} to map the support of $A$ to the \nameref{def:majorana} and therefore synthesize it as a \nameref{def:matchgate}.
\end{proof}

\begin{remark}\label{remark:quantum-polar}
One can avoid classically rounding to the nearest unitary if it suffices to implement a CPTP \emph{quantum channel} that approximates the unknown unitary. 
Let $A = \sum_{x \in \supp(A)} \alpha_x W_x$ be the output of our algorithm, which is not necessarily unitary.
Using \cref{lemma:lcu}, one can efficiently build a block-encoding $U_A$ of $\frac{A}{\sum_{x \in \supp(A)} \abs{\alpha_x}}$ and then apply \cite[Corollary 1, Eq.\ 8]{quek2022fast} to approximately apply the polar decomposition of $A$. 
(Note that scaling $A$ by a positive constant does not affect the polar decomposition.) 
This yields an implementation within diamond distance $\eps$ using $O\left(\frac{1}{\kappa}\log(1/\eps)\right)$ queries to $U_A$ and $U_A^\dagger$, where $1/\kappa$ is the smallest singular value of $\frac{A}{\sum_{x \in \supp(A)} \abs{\alpha_x}}$.
Since $A$ is $\eps$-close to a unitary, we have $1/\kappa \geq \frac{1-\eps}{\sum_{x \in \supp(A)} \abs{\alpha_x}} \geq \frac{1-\eps}{\sqrt{s}}$.
So for $\eps \leq 1/2$, the query complexity to $U_A$ is $O(\sqrt{s}\log(1/\eps))$.

Finally, this polar decomposition approach can also applies to the framework we present in \cref{sec:qnc-plus-clifford}: if it suffices to implement a quantum channel close to $U^\dagger \otimes U$, the same technique yields an efficient implementation. In particular, applying this technique to \cref{thm:efficient-learning-thm,cor:efficient-learning-thm} yields an algorithm that implements a quantum channel close to $U^\dagger \otimes U$ in diamond distance.
\end{remark}

\section{A Framework for Learning Structured Quantum Circuits}
\label{sec:qnc-plus-clifford}

\subsection{The Framework}
\label{subsec:framework}

We present our framework that yields efficient learning algorithms for structured unitaries. In particular, we identify a sufficient condition that implies efficient learning algorithms. 
To state this condition precisely, we must define the notion of a generating set.

\begin{definition}[Pauli Generating Set]\label{def:pauli-generating-set}
    A subset $G \subset \mathcal{P}_n$ of the $n$-qubit Pauli group is called a \emph{generating set} if every Pauli operator $P \in \{I, X, Y, Z\}^{\otimes n}$ can be expressed as a product of elements from $G$, up to a global phase. 
    Formally, for each such $P$, there exist a phase $c_P \in \{\pm 1, \pm i\}$ and a sequence of $\ell_P \le L$ generators $g_1, \dots, g_{\ell_P} \in G$ such that $P = c_P \prod_{j=1}^{\ell_P} g_j$. The minimum such upper bound $L$ is called the \emph{length} of $G$.
\end{definition}

We will prove the following. Let $U$ be an unknown unitary channel, and suppose we have a known generating set $G$ of Pauli operators for which $U^\dagger P U$ is $k$-Pauli-dimensional (resp. $s$-sparse) for all $P \in G$. Then there is an efficient learning algorithm that learns $U$ to diamond distance $\eps$ using $\poly(2^k, n)\cdot \frac{\log1/\delta}{\eps}$ (resp. $\poly(s, n) \cdot \frac{\log 1/\delta}{\eps^2}$) queries and time.
In the most general terms, if the Heisenberg evolution of a generating set is efficiently learnable, then the unitary itself is efficiently learnable.

After we establish the above theorem, we will show how this framework lifts to learning an infinite hierarchy of unitary channels that contains several natural classes of quantum circuits as special cases.

First, we establish that learning the Heisenberg-evolution of the generating set is information-theoretically sufficient to determine the global unitary.

\begin{theorem}[Learning generating sets suffices]\label{thm:generating-set-suffices}
    Let $G$ be a Pauli generating set with length $L$. Let $U$ and $V$ be $n$-qubit unitary channels. 
    If $V$ matches the Heisenberg evolution of $U$ on every generator $P \in G$ to operator norm accuracy $\eps'$:
    \[
        \opnorm{V^\dagger P V - U^\dagger P U} \le \eps',
    \]
    then the diamond distance between the unitary channels is strictly bounded by
    \[
        \distdiamond(U, V) \le 4 L \eps'.
    \]
\end{theorem}
\begin{proof}
    Let $W = V U^\dagger$. By assumption, for all $P \in G$,
    \[
        \opnorm{ W^\dagger P W - P}
        = \opnorm{ U V^\dagger P V U^\dagger - U U^\dagger P U U^\dagger} 
        = \opnorm{V^\dagger P V - U^\dagger P U} \le \eps'.
    \]
    By the unitary invariance of the operator norm, multiplying by $W$ on the left bounds the commutator: $\| W g - g W \|_\infty \le \eps'$.
    
    By \cref{def:pauli-generating-set}, any $n$-qubit Pauli operator $P \in \{I, X, Y, Z\}^{\otimes n}$ can be written as a product of at most $\ell \le L$ generators, $P = c_P \prod_{j=1}^{\ell} g_j$.
    We use a telescoping sum to bound the commutator of $W$ with any Pauli $P$. Since $\opnorm{g_j}=1$, we have
    \[
        \opnorm{W P - P W} \le \sum_{j=1}^{\ell} \opnorm{ W g_j - g_j W} \le \ell \eps' \le L \eps'.
    \]
    
    To bound the global distance between the channels, we apply a standard Pauli twirling identity: 
    \[
    \frac{1}{4^n} \sum_{P \in \{I,X,Y,Z\}^{\otimes n}} P W P = \frac{\Tr(W)}{2^n} I \eqqcolon \alpha I.\]
    Specifically, we rewrite the difference between $W$ and its identity component strictly in terms of the commutators. 
    \[
        \alpha I - W = \frac{1}{4^n} \sum_{P \in \{I, X, Y, Z\}^{\otimes n}} \left(P W P - W\right) = \frac{1}{4^n} \sum_{P \in \{I,X,Y,Z\}^{\otimes n}} (P W - W P) P.
    \]
    Taking the operator norm and applying the triangle inequality, we have 
    \[
        \opnorm{W - \alpha I} \le \frac{1}{4^n} \sum_{P \in \{I,X,Y,Z\}^{\otimes n}} \opnorm{(P W - W P) P} = \frac{1}{4^n} \sum_{P \in \{I,X,Y,Z\}^{\otimes n}} \opnorm{W P - P W} \le L \eps'.
    \]
    
    Because $W$ is unitary, $\opnorm{W} = 1$. By the reverse triangle inequality, $1 = \opnorm{W} \le \opnorm{\alpha I} + \opnorm{W - \alpha I} = |\alpha| + \opnorm{W - \alpha I}$, forcing $|\alpha| \ge 1 - L \eps'$.
    Setting $\alpha \coloneqq |\alpha| e^{i\theta}$, we factor out the global phase:
    \[
        \opnorm{e^{-i\theta} W - |\alpha| I} = \opnorm{ W - \alpha I } \le L \eps'.
    \]
    Thus, we can bound the distance from $W$ to the identity matrix:
    \[
        \opnorm{ e^{-i\theta} W - I} \le \opnorm{ e^{-i\theta} W - |\alpha| I }+ \big| |\alpha| - 1 \big| \le L \eps' + L \eps' = 2 L \eps'.
    \]
    
    The diamond distance between $U$ and $V$ is equal to the diamond distance between $W$ and the identity. Thus, we conclude the proof by applying \cref{fact:op-to-diamond}. 
\end{proof}

It is easy to see that $L \leq 2n$ for any generating set $G$. This is because Pauli operators (up to phase) can be identified with elements of $\F_2^{2n}$, so any such operator can always be expressed with $2n$ generating elements.

\cref{thm:generating-set-suffices} establishes that learning the Heisenberg-evolved operators of a generating set information-theoretically suffices to determine an unknown unitary operator.
However, to obtain an algorithm, we must also specify how to construct a description of the unknown unitary operator from the learned Heisenberg-evolved operators. 
We present such a construction now, which generalizes the approach of Huang, Liu, Broughton, Kim, Anshu, Landau, and McClean~\cite{huang2024shallow}.
The runtime of the resulting algorithm will depend on the number of generators needed to express the weight-$1$ Pauli operators, which we call the \emph{local length} of a generating set.

\begin{definition}[Local length of a generating set]
\label{def:local-length}
Let \(G\subseteq \mathcal{P}_n\) be a generating set. Define the \emph{local length} of $G$ as
\[
\ell_G \coloneqq \max_{i\in[n],\,Q\in\{X_i,Y_i,Z_i\}}
\min\Bigl\{
t : \exists\, P_1,\dots,P_t\in G,\ \omega\in\{\pm1,\pm i\}
\text{ such that }
Q=\omega \prod_{j=1}^t P_j
\Bigr\}.
\]
\end{definition}

\begin{theorem}[Efficient unitary learning via local length generating sets]\label{thm:efficient-learning-thm}
    Let $G$ be a known generating set.
    Let $U$ be any $n$-qubit unitary channel such that for every generator $g \in G$, the conjugated operator $ U_g \coloneqq U^\dagger g U$ is learnable to operator distance $\eps$ with probability at least $1-\delta$ using $q(\eps, \delta)$ queries to $U_g$ and $t(\eps, \delta)$ time. 
    Then, given queries to $U$ and $U^\dagger$, there is an algorithm that outputs a $2n$-qubit unitary channel $V$ satisfying $\distdiamond(U^\dagger \otimes U, V) \leq \eps$ with probability at least $1-\delta$.
    The algorithm uses $O\left(m \cdot q(\eps/L, \delta/m)\right)$ queries and $O\left(m \cdot t(\eps/L, \delta/m)\right)$ queries time, where $L$ is the local length of $G$ (\cref{def:local-length}) and $m \coloneqq \abs{G}$.
\end{theorem}
\begin{proof}
    We begin by querying the unknown channel to learn the Heisenberg evolution of each generator $g \in G$. 
    Each generator's conjugation is learnable by assumption.
    We set the target accuracy to $\eps' = \tfrac{\eps}{6L}$ and the failure probability to $\delta' = \tfrac{\delta}{m}$.
    By the union bound, all $m$ generators are successfully learned with probability at least $1-\delta$. 
    
    We round each output to the nearest unitary via singular value decomposition (see e.g. the proof of \cref{lem:round-to-unitary}), incurring at most a factor-$2$ loss. This yields a unitary approximation $\widehat{V}_g$ for each generator satisfying $\opnorm{\widehat{V}_g - U^\dagger g U} \le \tfrac{\eps}{3nL}$.

    Next, for each single-qubit Pauli $P \in \{X_i, Y_i, Z_i\}_{i=1}^n$, we classically compute its unitary approximation $\widehat{V}_P = c_P \prod_{j=1}^{\ell_P} \widehat{V}_{g_j}$.
    Because the exact evolution distributes perfectly as $U^\dagger P U = c_P \prod_{j=1}^{\ell_P} \left(U^\dagger g_j U\right)$, and both $\widehat{V}_{g_j}$ and $U^\dagger g_j U$ are strictly unitary, a standard telescoping sum bounds the operator norm error of the product by the sum of individual errors:
    \[
        \opnorm{\widehat{V}_P - U^\dagger P U} \le \sum_{j=1}^{\ell_P} \opnorm{ \widehat{V}_{g_j} - U^\dagger g_j U} \le L \left(\frac{\eps}{3nL}\right) = \frac{\eps}{3n}.
    \]

    Recall the simple identity used in~\cite{huang2024shallow}:
    \begin{equation}\label{eq:swap}
        U^\dagger \otimes U = \left(U^\dagger \otimes I^{\otimes n}\right) \cdot \mathrm{SWAP} \cdot \left( U \otimes I^{\otimes n} \right)\cdot \mathrm{SWAP},
    \end{equation}
    where $U^\dagger \otimes U$ is a $2n$-qubit unitary, and $\SWAP = \prod_i \SWAP_i$, with $\SWAP_i$ denoting the two-qubit swap gate between qubit $i$ and qubit $i+n$. 
    Note that the order of the $\SWAP_i$ gates does not matter since they act on disjoint pairs of qubits. 
    Expanding this product, we can rewrite \cref{eq:swap} as 
    \begin{equation}\label{eq:double-system-identity}
        U^\dagger \otimes U = \left[\prod_{i=1}^n\left(U^\dagger \otimes I^{\otimes n}\right) \cdot  \mathrm{SWAP}_i \cdot \left( U \otimes I^{\otimes n} \right) \right]\cdot \mathrm{SWAP}.
    \end{equation}
    By \cref{fact:op-to-diamond,fact:dist-composition}, it therefore suffices to learn each term $O_i \coloneqq (U^\dagger \otimes I^{\otimes n}) \cdot  \mathrm{SWAP}_i \cdot \left( U \otimes I^{\otimes n} \right)$ to accuracy $\frac{\eps}{n}$ in operator norm to learn $U^\dagger \otimes U$ to accuracy $\eps$ in diamond distance. 
    
    We now substitute the local approximations $\widehat{V}_P$ into \cref{eq:double-system-identity}. The $i$-th cross-term expands as:
    \[
        O_i = \frac{1}{2}\left(I \otimes I + (U^\dagger X_i U) \otimes X_i + (U^\dagger Y_i U) \otimes Y_i + (U^\dagger Z_i U) \otimes Z_i\right).
    \]
    Let $\widehat{O}_i$ denote the classical estimate of this term using our estimates $\widehat{V}_{X_i}, \widehat{V}_{Y_i}, \widehat{V}_{Z_i}$. By the triangle inequality, 
    \[
        \opnorm{\widehat{O}_i - O_i} \le \frac{1}{2}\left( 0 + \frac{\eps}{3n} + \frac{\eps}{3n} + \frac{\eps}{3n} \right) = \frac{\eps}{2n}.
    \]
    
    As $\widehat{O}_i$ is not strictly unitary, we again round it to the nearest unitary $\widehat{W}_i$ via singular value decomposition, which incurs at most another factor-$2$ loss. Thus, the final operator norm error per tensor factor is bounded by $\tfrac{\eps}{n}$. 
    
    By \cref{fact:dist-composition} together with \cref{eq:double-system-identity}, the global product $\left(\prod_{i=1}^n \widehat{O}_i\right) \cdot \mathrm{SWAP}$ approximates $U^\dagger \otimes U$ to within an error of $n \cdot \tfrac{\eps}{n} = \eps$ in diamond distance. 
    Substituting the generator accuracy $\eps' = \tfrac{\eps}{6nL}$ and failure probability $\delta' = \tfrac{\delta}{m}$ into the base algorithm's bounds yields the claimed query and time complexities, understanding that we have to run the learning algorithm $m$ times.
\end{proof}

\begin{corollary}\label{cor:efficient-learning-thm}
    Let $G$ be a known generating set.
    Let $U$ be any $n$-qubit unitary channel such that for every generator $g \in G$, the conjugated operator $U^\dagger g U$ is $k$-Pauli dimensional (resp. $s$-Pauli sparse \emph{and efficiently roundable to a unitary}). 
    Then, given queries to $U$ and $U^\dagger$, there is an algorithm that outputs a $2n$-qubit unitary channel $V$ satisfying $\distdiamond(U^\dagger \otimes U, V) \leq \eps$ with probability at least $1-\delta$.
    The algorithm uses $\poly(n, 2^k, m, L, \log(1/\delta)/\eps$ (resp. $\poly(n,s, m, L)\log(1/\delta)/\eps^2$) queries and time, where $L$ is the local length of $G$ (\cref{def:local-length}) and $m \coloneqq \abs{G}$.
\end{corollary}
\begin{proof}
    Combine \cref{cor:pauli-dimension-bootstrap} (resp. \cref{thm:learning-sparsity}) with \cref{thm:efficient-learning-thm}.
\end{proof}

\begin{remark}[Explicit query complexity of \cref{cor:efficient-learning-thm}]
\label{remark:explicity-query-time}
   The algorithm uses $O\left(m n L 2^k k \frac{\log(m/\delta)}{\eps}\right)$ (resp. $O\left(m n^2 L^2 s \frac{s + \log(m/\delta)}{\eps^2}\right)$) queries.
\end{remark}

\subsection{Generalizing to an Infinite Hierarchy of Circuits}
\label{subsec:hierarchy}

We now generalize \cref{cor:efficient-learning-thm} to an infinite hierarchy of unitary circuits.
The overarching message of \cref{subsec:framework} is that efficiently learning $U^\dagger g U$ for all $g$ in a generating set $G$ suffices to efficiently learn $U$.
\cref{cor:efficient-learning-thm} instantiates this theorem in the case where the operators $U^\dagger g U$ are learnable in the case they are $k$-Pauli dimensional or $s$-Pauli sparse. 

We will show that this framework is closed under a natural recursive lifting. 
Let $\calD_0$ (resp. $\calS_0$) denote the class of $n$-qubit unitary circuits that are $k$-Pauli dimensional (resp. $s$-Pauli sparse and efficiently round-able to a unitary).
For $d \geq 1$, define 
\[
\calD_d \coloneqq \{ U : \text{$\exists$ known generating set $G$ such that $U^\dagger g U \in \calD_{d-1}$ for all $g \in G$} \}, 
\]
and define $\calS_d$ analogously.

\begin{theorem}\label{thm:hierarchy}
    Fix $d \geq 0$. Then every unitary in $U \in \calD_d$ (resp. $U \in \calS_d$) can be learned to diamond distance $\eps$ with success probability at least $1 - \delta$ using $\poly(2^k, n^{2d}, m, L)\cdot \frac{\log1/\delta}{\eps}$ (resp. $\poly(s, n^d, m, L)\cdot \frac{\log (1/\delta)}{\eps^2}$) queries and time.
\end{theorem}
\begin{proof}
The proof proceeds by induction on $d$. 
The base cases $d=0$ and $d=1$ follow from \cref{cor:pauli-dimension-bootstrap} (resp. \cref{thm:learning-sparsity}) and \cref{cor:efficient-learning-thm}.
Now suppose the claim holds for $d-1$, and let $U \in \calD_d$ (resp. $U \in \calS_d$).
By definition, there exists a known generating set $G$ such that for every $g \in G$, the unitary
\[
U_g \coloneqq U^\dagger g U \in \calD_{d-1} \quad (\text{resp. } U_g \in \calS_{d-1}).
\]
To reconstruct $U$, it suffices (by \cref{thm:efficient-learning-thm}) to learn each $U_g$ to diamond distance $\eps' = O(\frac{\eps}{nL_d})$ and $\delta' = O(\delta/m_d)$, repeated $m$ times for each $U_g$.
\end{proof}

\begin{remark}[Explicit query complexity of \cref{thm:hierarchy}]
\label{remark:explicity-query-time-hierarchy}
    Let $G_i$ be the generating set used to learn level $\calD_{i}$ (resp. $\calS_i$)  and let $L_i$ be the local length of $G_i$.
    Then define $m \coloneqq \prod_{i=1}^d \abs{G_i}$ and $L \coloneqq \prod_{i=1}^d L_i$.
    The algorithm uses $O\left(m n^d L 2^k k \frac{\log(m/\delta)}{\eps}\right)$ (resp. $O\left(m n^{2d} L^2 s \frac{s + \log(m/\delta)}{\eps^2}\right)$) queries.
\end{remark}

\section{Applications}
\label{sec:applications}

We now apply the framework developed in \cref{sec:qnc-plus-clifford} to obtain efficient learning algorithms for several concrete and well-studied classes of quantum circuits. 
Specifically, we obtain learning algorithms for near-Clifford circuits, quantum $k$-juntas~\cite{chen2023testing}, the Clifford hierarchy~\cite{low2009learning}, fermionic matchgate circuits, and the matchgate hierarchy~\cite{matchgate-hierarchy}. 
Here, near-Clifford circuits are those composed of Clifford gates and $O(\log n)$ single-qubit non-Clifford gates. 
We also obtain learning algorithms for compositions of shallow circuits with near-Clifford circuits, as well as compositions of fermionic matchgate circuits with Clifford circuits.

At a high level, these results follow by showing that each of these circuit classes satisfies the condition on Heisenberg-evolved generating sets in \cref{thm:efficient-learning-thm}. 
The main message of this section is that this condition provides a unifying principle that both recovers and extends a wide range of existing learning algorithms, while in several cases yielding improved guarantees.

\subsection{Quantum \texorpdfstring{$k$}{k}-Juntas}
\label{subsec:juntas}

We present our query-optimal algorithm for learning quantum $k$-juntas in diamond distance. We begin by recalling the definition of a quantum $k$-junta.

\begin{definition}[Quantum $k$-junta]
    A quantum $k$-junta unitary channel is a unitary channel that only acts non-trivially on $k$-qubits.
\end{definition}
In other words, up to permutation of qubits, it acts as $I^{\otimes n-k} \otimes U$, where $U$ is some arbitrary $k$-qubit unitary matrix.
A query-optimal tomography algorithm for quantum junta channels follows from \cref{cor:pauli-dimension-bootstrap-finegrain}.

\begin{corollary}[Query optimal junta learner without inverse]\label{cor:optimal-junta}
Let $U \in \C^{2^n \times 2^n}$ be a $k$-junta unitary.
There is a tomography algorithm that, given query access to $U \in \C^{2^n \times 2^n}$ as well as parameters $\delta, \eps > 0$, outputs an estimate $V$ satisfying $\distdiamond(U, V) \leq \eps$ with probability at least $1-\delta$. 
Moreover, the algorithm satisfies the following properties:
\begin{itemize}
    \item $V$ is a $k^\prime$ junta for $k^\prime \leq k$.
    \item The algorithm makes at most $O\left(4^k \frac{\log (1/\delta)}{\eps}\right)$ queries to $U$ (and only requires forward access, i.e., it does not require $U^\dagger$ or controlled-$U$).
    \item The algorithm runs in time \[
        \widetilde{O}\left(\left( 4^k\left(n+\frac{4^{k}}{\eps} \right)\log(1/\delta) + \left(n + 8^{k}\right)\log^2(1/\delta)\right)\right).
    \]
    \item The algorithm uses between $n$ and $2n-1$ additional qubits of space.
\end{itemize}
\end{corollary}
\begin{proof}
    A $k$-junta WLOG takes the form of $I^{\otimes n-k} \otimes U$, so the support resides within $\paulisupport_{k, 0}$.
    Apply \cref{cor:pauli-dimension-bootstrap-finegrain} without the inverse for Pauli dimension $2k$ and parameters $a=k$ and $b=0$.
    The resulting query complexity is just $O\left(4^k \frac{\log (1/\delta)}{\eps}\right)$.
    Note that, using the bounds from \cref{remark:time-complexity}, the time complexity would be the same (up to logarithmic factors) with or without access to $U^\dagger$, so we will only choose to use forward queries.
\end{proof}

\subsection{Composition of Shallow and Near-Clifford Circuits}

We now give an algorithm for learning unitary channels that can be expressed as the composition of a shallow circuit with a near-Clifford circuit (in either order). 
By a shallow circuit we mean a depth-$d$ quantum circuit with arbitrary one- and two-qubit gates, and by a near-Clifford circuit we mean a Clifford circuit augmented with at most $t$ arbitrary single-qubit gates. 
Our algorithm learns such unitaries to accuracy $\eps$ in diamond distance in polynomial time, provided $d = O(\log \log n)$ and $t = O(\log n)$. 
The class of unitaries covered by our result includes the first level of the recently introduced \emph{Magic Hierarchy} \cite{parham2025quantumcircuitlowerbounds}.

\subsubsection{Clifford Nullity}

\emph{A priori}, it is not obvious which natural classes of unitary channels (besides shallow depth circuits) satisfy the conditions of \cref{cor:efficient-learning-thm}. 
In this subsection, we show that Clifford circuits doped with a small number of single-qubit non-Clifford gates do satisfy these conditions, which in turn yields an alternative learning algorithm for this class.
In fact, in \cref{ssec:nullity-and-shallow} we will learn shallow circuits composed with the more general class of near-Clifford unitaries involving small Clifford nullity \cite{jiang2023lower}, which we define below.\footnote{We chose the name `Clifford nullity', rather than `unitary stabilizer nullity' from \cite{jiang2023lower}.}

\begin{definition}[Clifford nullity]\label{def:clifford-nullity}
    An $n$-qubit unitary channel $U \in \C^{2^n \times 2^n}$ has Clifford nullity $t$ if there exists a subspace $S \subseteq \F_2^{2n}$ of co-dimension $t$ such that for every $x \in S$, there exists some $y\in \F_2^{2n}$ satisfying $U^\dagger W_x U = \pm W_y$.
\end{definition}

Intuitively, Clifford nullity measures how much of the Pauli group is normalized by $U$: nullity $0$ corresponds exactly to Clifford circuits.
Importantly, if $S$ is the subspace from \cref{def:clifford-nullity}, then for all $x \in S$, the conjugate $U^\dagger W_x U$ lies in a subspace $S^\prime$ of co-dimension $t$.
Additionally, if $S$ admits a decomposition into an $2a$-dimensional symplectic part and a $b$-dimensional isotropic part then $S^\prime$ does as well.

To understand Clifford nullity, we establish the following structural result.

\begin{fact}\label{fact:clifford-nullity-structure}
    Let $U \in \C^{2^n \times 2^n}$ be an $n$-qubit unitary channel with Clifford nullity $t$.
    Then $U$ can be decomposed into $C_2 U^\prime C_1$ where $C_2$ and $C_1$ are Clifford unitary channels and $U^\prime$ is $2n-t$-Pauli dimensional.
    Furthermore, if the subspace $S$ that is stabilized by $U$ has $(n-a-b, b)$ decomposition into a symplectic part of dimension $2(n-a-b)$ and isotropic part of dimension $b$ (so that $2a+b = t$), then $U^\prime$ has Pauli support in $\paulisupport_{a, b}$.
\end{fact}
\begin{proof}
    For conciseness in this proof, we will ignore positive and negative signs when conjugating Pauli operators by a Clifford.
    Define subspace $T \coloneqq \paulisupport_{a,b}^\sympcomp$.\footnote{One should think of $T$ as the Paulis corresponding to $\{I, X, Y, Z\}^{\otimes(n-a-b)}\otimes \{I, Z\}^{b} \otimes I^{\otimes a}$.}
    Let $S^\prime \coloneqq U^\dagger S U$ be the image of conjugating the stabilized subspace $S$ by $U$.
    Let $C_1$ be a Clifford circuit such that $C_1^\dagger S C_1 = T$.
    Then let $C_2$ be a Clifford circuit such that $C_2^\dagger T C_2 = S^\prime$.
    Finally, we will stipulate that for $x \in S$, $C_2^\dagger C_1^\dagger W_x C_1 C_2 = U^\dagger W_x U$.
    
    We will now show that $U^\prime \coloneqq C_2^\dagger U C_1^\dagger$ commutes with any $W_x$ where $x \in T$.
    Observe that for $x \in \paulisupport_{n-a-b, b}^\sympcomp$, $C_1 W_x C_1^\dagger = W_y$ for $y \in S$.
    Therefore,
    \[
        (U^\prime)^\dagger W_x U^\prime = (C_2 U^\dagger C_1) W_x (C_1^\dagger U C_2^\dagger) = C_2 (U^\dagger W_y U) C_2^\dagger = C_2 (C_2^\dagger C_1^\dagger W_y C_1 C_2) C_2^\dagger = C_1^\dagger W_y C_1 = W_x.
    \]
    To simultaneously commute with everything in $T$, it follows that $\supp(U^\prime) \subseteq T^\sympcomp = \paulisupport_{a,b}$.
\end{proof}

It is not too difficult to show that the converse of \cref{fact:clifford-nullity-structure} is also true: any unitary with the form $U \coloneqq C_2 U^\prime C_1$ where $C_1$ and $C_2$ are Clifford unitary channels and $\supp(U^\prime) \subseteq \paulisupport_{a, b}$ such that $2a+b = t$.
As such, it completely characterizes Clifford nullity.
We also note that $k$-Pauli dimensional unitary channels are therefore a strict subset of Clifford nullity $k$ unitary channels, as we can take $C_2 = I = C_1$.

To better understand these Heisenberg-evolve Pauli operators, we will use the following fact about conjugation of a Pauli dimensional matrix by a Pauli dimensional matrix.

\begin{fact}\label{fact:pauli-dimension-conjugation}
    Let $A$ and $B$ be $k$- and $\ell$-Pauli dimensional matrices, respectively.
    Then $A^\dagger B  A$ is $(k+\ell)$-Pauli dimensional, with support over the subspace spanned by $\langle \supp(A), \supp(B) \rangle$.
\end{fact}
\begin{proof}
    We can observe that $A = \sum_{x \in G_1} \alpha_x W_x$ and $B = \sum_{y \in G_2} \alpha_y W_y$ for subspaces $G_1$ and $G_2$ with dimension $k$ and $\ell$ respectively.
    Then \begin{align*}
        A^\dagger B A = \sum_{x_1, x_2 \in G_1} \sum_{y \in G_2} \alpha_{x_1}^* \alpha_{x_2} \alpha_y W_{x_1} W_{y} W_{x_2} = \sum_{x_1, x_2 \in G_1} \sum_{y \in G_2} \alpha_{x_1}^* \alpha_{x_2} \alpha_y (-1)^{[y, x_2]} W_{x_1} W_{x_2} W_y
    \end{align*}
    As $x_1, x_2$ in a subspace $G_1$, we see that $\supp(A^\dagger B A)$ must lie in the subspace generated by $G_1$ and $G_2$, which is at most $(k+\ell)$-dimensional.
\end{proof}

\begin{remark}
    One should note that for two $k$ and $\ell$ Pauli dimensional matrices $A$ and $B$, $AB$ must also be $(k+\ell)$-Pauli dimensional.
    The benefit of \cref{fact:pauli-dimension-conjugation} is that $A^\dagger B A$ is still only $(k+\ell)$-Pauli dimensional, rather than the na\"ive $(2k+\ell)$-Pauli dimensional bound.
\end{remark}

Using \cref{fact:clifford-nullity-structure,fact:pauli-dimension-conjugation} we can now show that $\supp(U^\dagger A U)$ is low Pauli dimensional when $A$ is low Pauli dimensional and $U$ is a unitary channel with small Clifford nullity.

\begin{lemma}\label{lemma:nullity-conjugation}
    $U \in \C^{2^n \times 2^n}$ have Clifford nullity $t$ and let $A$ be a matrix that is $k$-Pauli dimensional.
    Then $U^\dagger A U$ is at most $(k+t)$-Pauli dimensional.
\end{lemma}
\begin{proof}
    From \cref{fact:clifford-nullity-structure}, $U = C_2 U^\prime C_1$ for Clifford unitary channels $C_2$ and $C_1$ and $t$-Pauli dimensional unitary channel $U^\prime$.
    Observe that 
    \begin{align*}
        U^\dagger A U = (C_2 U^\prime C_1)^\dagger A (C_2 U^\prime C_1)\\
        = C_1^\dagger \left((U^\prime)^\dagger (C_2^\dagger A C_2) U^\prime\right) C_1
    \end{align*}
    and that $B \coloneqq C_2^\dagger A C_2$ is still $k$-Pauli dimensional.
    Therefore $D \coloneqq (U^\prime)^\dagger B U^\prime$ is $(k+t)$-Pauli dimensional by \cref{fact:clifford-nullity-structure}.
    Finally, $C_1^\dagger D C_1$ is again still $(k+t)$-Pauli dimensional as the Clifford unitary just permutes the Pauli decomposition.
\end{proof}

Last but not least, to connect this definition to circuits one may be more familiar with, we recall a standard fact from the stabilizer formalism literature (see, e.g., \cite{leone-stabilizer-nullity,grewal2023improved}), which shows that Clifford circuits doped with a small number of single-qubit non-Clifford gates has small Clifford nullity.

\begin{fact}\label{fact:doped-to-nullity}
    A Clifford circuit augmented by $t$ single-qubit non-Clifford gates has Clifford nullity at most $2t$.\footnote{If the non-Clifford gates are limited to single-qubit Pauli rotations, such as the $T$ gate, then the nullity is at most $t$.}
\end{fact}

This is actually a special case of the more general fact that alternations of juntas and Clifford circuits have bounded Clifford nullity.
It  follows that such circuits can also be efficiently learned by us, even in composition with a shallow-depth circuit.

\begin{fact}\label{fact:junta-to-nullity}
    Let $U$ be a circuit that alternates between unitaries with nullity $t_i$ and quantum juntas on $k_i$ qubits.
    Then $U$ has Clifford nullity at most $\sum_i 2k_i + t_i$.
\end{fact}

\subsubsection{Tomography Algorithm}\label{ssec:nullity-and-shallow}

We now establish the key technical fact: for $U$ decomposable into a shallow circuit and a unitary of bounded Clifford nullity, the Heisenberg-evolved single-qubit Paulis remain low-dimensional.

\begin{lemma}\label{lem:qnc-plus-clifford-heisenberg-pauli}
    Let $U \in \C^{2^n \times 2^n}$ be expressible as $U = QC$, where $Q$ is a depth-$d$ quantum circuit and $C$ has Clifford nullity $t$.
    Then, for every weight-one Pauli operator $P_i$ that acts non-trivially only on qubit $i$, $U^\dagger P U$ is $(2^{d+1} + t)$-Pauli dimensional.
    Moreover, $\supp(C_2 U^\dagger P U C_2^\dagger) \subseteq \paulisupport_{2^{d}+\ell, t-2\ell} $ for some \emph{other} Clifford circuit $C_2$ and integer $\ell \leq \lfloor \frac{t}{2}\rfloor$.
\end{lemma}
\begin{proof}
    Observe (as we argued earlier in this section) that $Q^\dagger P_i Q$ is a $2^d$-junta. Because $k$-juntas are $2k$-Pauli dimensional, the Heisenberg-evolved operator is therefore $2^{d+1}$-Pauli dimensional with a $(2^d, 0)$ structure.
    Finally, apply \cref{lemma:nullity-conjugation} to show that $C^\dagger (Q^\dagger P_i Q) C$ is $(2^{d+1}+t)$-Pauli dimensional.
\end{proof}

We are now ready to state our learning algorithm for unitary channels that decompose into a shallow circuit composed with a unitary of bounded Clifford nullity.
Together with \cref{fact:doped-to-nullity}, which shows that Clifford circuits augmented with $t$ single-qubit gates have Clifford nullity $O(t)$, this yields the main result of the section.

\begin{theorem}\label{thm:learn-magic-heirarchy}
    Let $U \in \C^{2^n \times 2^n}$ be an $n$-qubit unitary that can be written either as $U = QC$ or $U = CQ$, where $Q$ is a depth-$d$ circuit and $C$ has Clifford nullity $t$.
    Then, given query access to $U$ and $U^\dagger$, there exists an algorithm that, with probability at least $1-\delta$, outputs a $2n$-qubit unitary channel $V$ such that $\distdiamond(U^\dagger \otimes U, V) \leq \eps$ using \[O\left(2^{2^d + t}\left(2^{2^d} + t\right) \, \tfrac{n^2}{\varepsilon} \log(n/\delta)\right)\]
    queries to $U$ and $U^\dagger$, and 
    \[
        \widetilde{O}\left(\left(2^{3\cdot 2^d + t} \left(8^t + \log(n/\delta)\right) + n\left(\frac{2^{4\cdot 2^d +2t}}{\eps}  + \log(n/\delta)\right)\right) n \log(n/\delta)\right)
    \]
    time.
\end{theorem}
\begin{proof}
Since we have access to both $U$ and $U^\dagger$, we may assume without loss of generality that $U = Q C$.
By \cref{lem:qnc-plus-clifford-heisenberg-pauli}, for every weight-one Pauli $P_i$, the conjugate $U^\dagger P_i U$ is $(2^{d+1}+t)$-Pauli dimensional.
Furthermore, for every weight-one Pauli $P_i$, there exists a Clifford circuit $C_2$ and integer $t \leq \lfloor \frac{t}{2}\rfloor$ where $\supp(C_2 U^\dagger P_i U C_2^\dagger) \subseteq \paulisupport_{2^d + \ell, t-\ell}$.
Applying \cref{cor:efficient-learning-thm} with worst-case parameter $a=2^d$ and $b = t$ then yields the claimed query and time complexities.
\end{proof}

Let us conclude with a few remarks about \cref{thm:learn-magic-heirarchy}.
First, our algorithm is \emph{improper} in the sense that the output is not itself a shallow circuit composed with a near-Clifford circuit. 
Instead, the algorithm returns a circuit consisting of $\widetilde{O}(2^d + t)$ alternating layers of Clifford circuits and depth-$\widetilde{O}(4^{2^{d}+t})$ circuits.
We leave it as an open problem to design a proper learning algorithm. 

Second, the $1/\eps$ Heisenberg scaling from \cref{cor:pauli-dimension-bootstrap-finegrain} results in an improved $n$ dependence in \cref{thm:learn-magic-heirarchy} for this concept class.
Had our algorithm scaled as $1/\eps^2$, there would be an extra factor of $n$ in the query and time complexities of \cref{thm:learn-magic-heirarchy} wherever $\eps$ appears.

Third, notion of Clifford nullity is quite powerful, as shown in \cref{fact:junta-to-nullity}.
This means that the following can be learned as a corollary of \cref{thm:learn-magic-heirarchy}.

\begin{corollary}\label{cor:alternating-junta-clifford}
    Let $U \in \C^{2^n \times 2^n}$ be an $n$-qubit unitary that can be written as
    \[
        U \coloneqq \prod_i (J_i C_i)
    \]
    where $J_i$ is a junta on $k_i$ qubits and $C_i$ is a unitary with Clifford nullity $t_i$.
    Then, given query access to $U$ and $U^\dagger$, there exists an algorithm that, with probability at least $1-\delta$, outputs a $2n$-qubit unitary channel $V$ such that $\distdiamond(U^\dagger \otimes U, V) \leq \eps$ using \[O\left(2^{t} t \, \tfrac{n^2}{\varepsilon} \log(n/\delta)\right)\]
    queries to $U$ and $U^\dagger$, and 
    \[
        \widetilde{O}\left(\left(2^{t} \left(8^t + \log(n/\delta)\right) + n\left(\frac{2^{2t}}{\eps}  + \log(n/\delta)\right)\right) n \log(n/\delta)\right)
    \]
    time, where $t = \sum_i 2k_i + t_i$.
\end{corollary}
\begin{proof}
    By \cref{fact:junta-to-nullity}, we can see that the Clifford nullity of $U$ is at most $t$.
    We then apply \cref{thm:learn-magic-heirarchy}.
\end{proof}

\subsection{Composition of Matchgate and Clifford Circuits}
Like Clifford circuits, matchgates are a set of highly expressive, but non-universal circuits that are classically simulable~\cite{valiant2002quantum,bravyi2019approximation}.
They are equivalent to fermionic Gaussian Unitaries \cite{terhal2002classical,knill2001fermioniclinearopticsmatchgates}, which model non-interacting fermions, after undergoing the Jordan--Wigner transformation to map them to a qubit system.
They are therefore important in condensed matter physics, quantum chemistry, and many-body physics.

\begin{definition}[Jordan--Wigner Majoranas]\label{def:majorana}
    We define the Majorana operators on a qubit system, having undergone the Jordan--Wigner transformation, to be 
    \[
    \gamma_{2a-1} \coloneqq Z^{\otimes a-1} \otimes X \otimes I^{\otimes n-a},\,\, \gamma_{2a} \coloneqq Z^{\otimes a} \otimes Y \otimes I^{\otimes n-a}
    \]
    for $a \in [n]$.
\end{definition}

Observe that the Jordan--Wigner Majoranas are a generating set of minimal size $2n$, but maximal local length $2n$.
They are also mutually anti-commuting, which we will need to apply \cref{fact:majorana-rounding}.

\begin{definition}[Matchgate circuit]\label{def:matchgate}
    We define the set of Match gates circuits to be the circuits such that for all $a \in [2n]$
    \[
    U^\dagger \gamma_a U = \sum_{\ell = 1}^{2n} M_{b a} \gamma_b
    \]
    for some orthogonal matrix $M \in O(2n)$.
\end{definition}

From the definition, we can see that the Heisenberg-evolved Jordan--Wigner Majoranas of a Match gate circuit are $2n$-sparse.

\begin{theorem}\label{thm:learn-match-plus-clifford}
    Let $U \in \C^{2^n \times 2^n}$ be an $n$-qubit unitary that can be written either as $U = MC$ or $U = CM$, where $M$ is a Matchgate circuit and $C$ is a Clifford circuit.
    Then, given query access to $U$ and $U^\dagger$, there exists an algorithm that, with probability at least $1-\delta$, outputs a $2n$-qubit unitary channel $V$ such that $\distdiamond(U^\dagger \otimes U, V) \leq \eps$ using \[O\left(n^6 \frac{n + \log(1/\delta)}{\eps^2}\right)\]
    queries to $U$ and $U^\dagger$, and 
    \[
        O\left(n^7 \frac{n + \log(1/\delta)}{\eps^2}\right)
    \]
    time.
\end{theorem}
\begin{proof}
    Observe that conjugating the Jordan--Wigner Majorana operator by a Matchgate circuit leaves us with a $2n$-Pauli sparse unitary composed of mutually anti-commuting Pauli operators.
    By then conjugating with a circuit with Clifford circuit, these anti-commuting Pauli operators will get mapped to a different set of mutually anti-commuting Pauli operators, whilst preserving sparsity.
    This means that we can time efficiently round the results of \cref{thm:learning-sparsity} to a unitary via \cref{fact:majorana-rounding}.
    We then apply \cref{cor:efficient-learning-thm} with $\abs{G} = 2n$, the number of Majoranas operators, Pauli sparsity $s  = O(2^t n)$, and $L = 2n$ the local length.
\end{proof}

\subsection{The Clifford and Matchgate Hierarchies}

We now show that our techniques yields learning algorithms for both the Clifford hierarchy~\cite{gottesman1999demonstrating} and the recently introduced matchgate hierarchy~\cite{matchgate-hierarchy}. 
Learning algorithms for these hierarchies were previously given in ~\cite{low2009learning,matchgate-hierarchy}.
First, we define the hierarchies. 

\begin{definition}[Clifford hierarchy]
    For $n \in \mathbb{N}$, let $\calC_0$ be the set of Pauli operators.
    Then for $k \geq 1$, we define $\calC_{k} \coloneqq \{U \in U(2^n) : \forall P \in \calC_0, U^\dagger P U \in \calC_{k-1}\}$ to be the $k$-the level of the Clifford Hierarchy.
\end{definition}

\begin{definition}[Matchgate hierarchy]
    For $n \in \mathbb{N}$, let $\calC_0$ be the set of unitaries supported solely on the Jordan--Wigner Majorana operators.
    Then for $k \geq 1$, we define $\calC_{k} \coloneqq \{U \in U(2^n) : \forall P \in \calC_0, U^\dagger P U \in \calC_{k-1}\}$ to be the $k$-the level of the matchgate hierarchy.
\end{definition}

It is immediately evident from their recursive definitions that both of these hierarchies lie in the hierarchy defined in \cref{subsec:hierarchy} and are therefore efficiently learnable by \cref{thm:hierarchy}.

\section{Lower Bounds}\label{sec:lower-bounds}
In this section, we prove query lower bounds for various unitary learning tasks.
In \cref{subsec:lower-bounds-dimenstionality-etc}, we prove lower bounds for learning low-Pauli-dimensional and Pauli-sparse unitaries, as well as quantum juntas.  
In \cref{subsec:lower-bounds-composition}, we discuss lower bounds for learning unitaries that can be expressed as the composition of near-Clifford unitaries and shallow circuits.

\subsection{Lower Bounds for Pauli Dimensionality, Sparsity, and Quantum Juntas}
\label{subsec:lower-bounds-dimenstionality-etc}
We prove lower bounds for learning quantum $k$-juntas, $s$-Pauli-sparse, and $k$-Pauli dimensional unitary channels. 
Our lower bounds leverage \cite[Theorem 1.2]{haah2023query}, which showed that $\Omega(d^2/\eps)$ queries are necessary to learn $d \times d$ unitary matrices, along with padding arguments.
These lower bounds establish the query optimality of our \cref{cor:pauli-dimension-bootstrap} and \cref{cor:optimal-junta}. For sparsity, we prove an $\Omega(s/\eps)$ lower bound, so a gap remains between that and our $O(s^2/\eps^2)$ upper bound in \cref{thm:learning-sparsity}.

\begin{lemma}[{\cite[Theorem 1.2]{haah2023query}}]\label{lem:unitary-lowerbound}
Let $\calA$ be an algorithm that, for an unknown unitary
$U \in \C^{d \times d}$ accessible through black box oracles that implement $U$, $U^\dagger$, $cU = \ketbra{0}{0} \otimes I_d + \ketbra{1}{1} \otimes U$,
and $cU^\dagger = \ketbra{0}{0} \otimes I_d + \ketbra{1}{1} \otimes U^\dagger$, can output a classical description of a unitary $V$
such that $\distdiamond(U, V) < \eps < \frac{1}{8}$ with probability $\geq \frac{2}{3}$.
Then $\calA$ must use $\Omega(d^2/\eps)$ oracle queries.
\end{lemma}

\begin{theorem}\label{thm:junta-lowerbound}
    Let $\calA$ be an algorithm that, for an unknown $k$-junta
    $U \in \C^{2^n \times 2^n}$ accessible through black box oracles that implement $U$, $U^\dagger$, $cU = \ketbra{0}{0} \otimes I_d + \ketbra{1}{1} \otimes U$,
    and $cU^\dagger = \ketbra{0}{0} \otimes I_d + \ketbra{1}{1} \otimes U^\dagger$, can output a classical description of a unitary $V$
    such that $\distdiamond(U, V) < \eps < \frac{1}{8}$ with probability $\geq \frac{2}{3}$.
    Then $\calA$ must use $\Omega(4^k/\eps)$ oracle queries.
\end{theorem}
\begin{proof}
    It's clear from \cref{lem:unitary-lowerbound} that the set of unitaries on $k$-qubits requires $\Omega(4^k/\eps)$ queries.
    If we take those unitaries and create a set of $n$-qubit unitaries by padding $I^{\otimes n-k}$ to the first (or last) register, then we get a set of $k$-junta.
    Learning this set of $k$-junta using $o(4^k/\eps)$ queries would contradict \cref{lem:unitary-lowerbound}, as querying $I^{\otimes n-k} \otimes U$ is just as easy as querying $U$ (modulo the $n-k$ extra ancilla qubits).
\end{proof}

\begin{corollary}\label{cor:dimension-lowerbound}
    Let $\calA$ be an algorithm that, for an unknown $k$-Pauli dimensional unitary
    $U \in \C^{2^n \times 2^n}$ accessible through black box oracles that implement $U$, $U^\dagger$, $cU = \ketbra{0}{0} \otimes I_d + \ketbra{1}{1} \otimes U$,
    and $cU^\dagger = \ketbra{0}{0} \otimes I_d + \ketbra{1}{1} \otimes U^\dagger$, can output a classical description of a unitary $V$
    such that $\distdiamond(U, V) < \eps < \frac{1}{8}$ with probability $\geq \frac{2}{3}$.
    Then $\calA$ must use $\Omega(2^k/\eps)$ oracle queries.
\end{corollary}
\begin{proof}
    Every $k$-junta is $2k$-Pauli dimensional, so an algorithm for $2k$-Pauli dimensional unitaries that runs in time $o(2^{(2k)}/\eps) = o(4^k/\eps)$ would violate \cref{thm:junta-lowerbound}.
\end{proof}

We can also show an $\Omega(s/\eps)$ lower bound for Pauli sparsity.

\begin{corollary}
\label{cor:pauli-sparsity-LB}
    Let $\calA$ be an algorithm that, for an unknown $s$-Pauli sparse unitary
    $U \in \C^{2^n \times 2^n}$ accessible through black box oracles that implement $U$, $U^\dagger$, $cU = \ketbra{0}{0} \otimes I_d + \ketbra{1}{1} \otimes U$,
    and $cU^\dagger = \ketbra{0}{0} \otimes I_d + \ketbra{1}{1} \otimes U^\dagger$, can output a classical description of a unitary $V$
    such that $\distdiamond(U, V) < \eps < \frac{1}{8}$ with probability $\geq \frac{2}{3}$.
    Then $\calA$ must use $\Omega(s/\eps)$ oracle queries.
\end{corollary}
\begin{proof}
    For $s = 2^k$ for some integer $k$, we can see that the set of $k$-Pauli dimensional unitary channels requires $\Omega(2^k/\eps) = \Omega(s/\eps)$ queries.
    As a $k$-Pauli dimensional unitary channel is also $2^k$-Pauli sparse, the lower bound follows.
\end{proof}

\subsection{Lower Bounds for the Composition of Shallow and Near-Clifford Circuits}
\label{subsec:lower-bounds-composition}

We now prove lower bounds for learning compositions of near-Clifford circuits with shallow circuits. 
An $\Omega(2^t)$ dependence in the query complexity of \cref{thm:learn-magic-heirarchy} is unavoidable in general, since unitary channels with Clifford nullity $2t$ form a strict superset of quantum $t$-juntas, and thus the lower bound of \cref{thm:junta-lowerbound} applies. 
However, it remains open whether the same lower bound holds in the more restricted setting of Clifford circuits augmented with $t$ single-qubit non-Clifford gates. 
This situation is analogous to the state-learning lower bounds discussed in \cite{grewal2023efficient}.

We also give a short proof that $\exp(\exp(\Omega(d)))$ query complexity is necessary even when inverse queries are allowed. 
This follows from the lower bound of \cite{bennet1997strengths} for Grover search via a padding argument. 
The corresponding non-padded argument appears in \cite[Proposition 3]{huang2024shallow}, which shows a weaker $\Omega(\exp(n))$ lower bound for $d = O(\log n)$.

\begin{lemma}
    Circuits of depth $d$ require $\exp(\exp(\Omega(d)))$ queries to learn to diamond distance $\eps < 1$ with constant success probability.
\end{lemma}
\begin{proof}
    The multi-controlled Toffoli gate with $k$ control qubits has been shown to be implemented by circuits using depth $d = O(\log k)$.\footnote{This is equivalent to the statement that $\mathsf{QAC^0} \subset \mathsf{QNC}^1$.}
    This multi-controlled Toffoli can be used to build the following family of unitary channels $\{U_y\}_{y \in \{0, 1\}^k}$ on $k$-qubits using two extra layers of Pauli $X$ gates:
    \[ U_y\ket{x} = \begin{cases}
        (-1)\ket{x} & x=y\\
        \ket{x} & x\neq y
    \end{cases}\]
    meaning that this gate can also be implemented in depth $d = O(\log k)$.
    Deciding if an unknown $k$-qubit unitary channel is the $k$-qubit identity matrix or one of $U_y$ (for all $y \in \{0, 1\}^k$) is equivalent to computing the AND function on $2^k$ bits.
    This provably requires $\Omega(2^{k/2}) = 2^{\exp(\Omega(d))}$ queries to achieve constant success probability \cite{bennet1997strengths}.
    
    Each $U_y$ is maximally far from identity in diamond distance (i.e., $\distdiamond(I^{\otimes k}, U_y) = 2$) by reducing to distinguishing the $k+1$-qubit state $\ket{\psi_y} \coloneqq \frac{\ket{y} + \ket{y \oplus 0\dots0 1}}{\sqrt{2}} = I^{\otimes k} \ket{\psi_y}$ from the orthogonal state $\frac{-\ket{y} + \ket{y \oplus 0\dots0 1}}{\sqrt{2}} = U_y \ket{\psi_y}$.
    Therefore, learning to diamond distance strictly less than $1$ allows one to distinguish the identity channel from the $U_y$.
    It follows that such an algorithm must use $\exp(\exp(\Omega(d)))$ queries even with inverse-access, as $U_y = U_y^\dagger$ and $I = I^\dagger$.
\end{proof}

\section*{Acknowledgements}
We thank Nick-Hunter Jones, Vishnu Iyer, William Kretschmer, Ewin Tang, and Fang Song for useful discussions.
DL is supported by US NSF Award CCF-222413.
SG is supported in part by an IBM PhD Fellowship.
This work was done in part while SG was visiting the Simons Institute for the Theory of Computing, supported by NSF Grant QLCI-2016245.

\bibliographystyle{alphaurl}
\bibliography{refs}

\appendix

\section{\texorpdfstring{Proof of \cref{lem:symplectic-gram-schmidt}}{Proof of Lemma 5.6}}\label{sec:deferred-proof}

\symplecticgramschmidt*
\begin{proof}
    The algorithm runs in three main phases that are each akin to Gram-Schmidt, with minor variations between each one.
    The first phase gets the generators of $T$, the second grabs the $x_{a^\prime + 1}, \dots, x_{a^\prime+k}$ generators of $S$, and the third grabs the remaining generators of $S$.
    This is in order to properly preserve the generators of $T$ by not mixing with generators of $S$.

    Instantiate counter $\ell \gets 1$ and set $A \gets \emptyset$.
    Let $G \gets \{t_1, \dots, t_c\}$ where $T = \langle t_1, \dots, t_d \rangle$ are the generators we have as input.
    Let $H \gets \{s_1, \dots, s_{d-c}\}$ be the additional generators of $S$.

    We then repeat the following until $G$ is the empty set \textbf{(phase one)}:
    \begin{itemize}
    \item Grab arbitrary $t_i \in G$ and remove it from $G$.
    \item Iterate through the rest of $G$ to find $t_j$ such that $[t_i, t_j] = 1$, should it exist.
    \item If such an $t_j$ \emph{does} exist:
    \begin{itemize}
        \item Remove $t_j$ from $G$ as well
        \item Label $x_\ell \gets t_j$ and $z_\ell \gets t_i$
        \item Iterate through the remaining $t_k \in G$ and if $[x_\ell, t_k] = 1$ then $t_k \gets t_k + z_\ell$ and if $[z_\ell, t_k] = 1$ then $t_k \gets t_k + x_\ell$.
        \item Iterate through $s_k \in H$ and if $[x_\ell, s_k] = 1$ then $s_k \gets s_k + z_\ell$ and if $[z_\ell, s_k] = 1$ then $s_k \gets s_k + x_\ell$.
        \item Increment $\ell \gets \ell + 1$.
    \end{itemize}
    \item If such an $t_j$ \emph{does not} exist then add $t_i$ to $A$.
    \end{itemize}
    
    We now repeat this next process until $A$ is the empty set \textbf{(phase two)}:
    \begin{itemize}
        \item Grab arbitrary $t_i \in A$ and remove it from $A$.
        \item Iterate through $H$ to find $s_j$ such that $[t_i, s_j] = 1$.
        \item If such an $s_j$ \emph{does} exist:
        \begin{itemize}
            \item Remove $s_j$ from $H$
            \item Label $x_\ell \gets s_j$ and $z_\ell \gets t_i$
            \item Iterate through $t_k \in A$ and if $[x_\ell, t_k] = 1$ then $t_k \gets t_k + z_\ell$.
            \item Iterate through $s_k \in H$ and if $[x_\ell, s_k] = 1$ then $s_k \gets s_k + z_\ell$ and if $[z_\ell, s_k] = 1$ then $s_k \gets s_k + x_\ell$.
            \item Increment $\ell \gets \ell + 1$.
        \end{itemize}
        \item If such an $s_j$ \emph{does not} exist then add $t_i$ to $G$ (recall that $G$ was emptied earlier).
    \end{itemize}
    Set elements of $G$ to $z_{a^\prime+k+1}, \dots, z_{a^\prime + b^\prime}$.
    
    Finally, we repeat this last process until $H$ is the empty set \textbf{(phase three)}:
    \begin{itemize}
        \item Grab arbitrary $s_i \in H$ and remove it from $H$.
        \item Iterate through $H$ to find $s_j$ such that $[s_i, s_j] = 1$.
        \item If such an $s_j$ \emph{does} exist:
        \begin{itemize}
            \item Remove $s_j$ from $H$
            \item Label $x_\ell \gets s_j$ and $z_\ell \gets s_i$
            \item Iterate through $s_k \in H$ and if $[x_\ell, s_k] = 1$ then $s_k \gets s_k + z_\ell$ and if $[z_\ell, s_k] = 1$ then $s_k \gets s_k + x_\ell$.
            \item Increment $\ell \gets \ell + 1$.
        \end{itemize}
        \item If such an $s_j$ \emph{does not} exist then add $s_i$ to $A$ (recall that $A$ was emptied earlier).
    \end{itemize}
    At the very end, set the generators in $A$ to be $z_{a+b^\prime-k+1}, \dots z_{a+b}$.

    Since we only add generators to generators, we still have a set of generators for $S$.
    Importantly, since we only ever add generators of $T$ to generators of $T$, we also still have generators for $T$.
    
    To satisfy the symplectic product relations, we note that after the first phase $A$ contains generators that commute will all other generators within $T$, so the (future) $z_{a^\prime+1}, \dots, z_{a^\prime + b^\prime}$ satisfy all of their requirements \emph{within $T$}.
    For $x_1, z_1$, we can see that $[x_1, z_1] = 1$, as we do not touch them once set.
    Furthermore, we force all remaining basis elements to commute with both $x_1$ and $z_1$.
    This includes all future $x_i, z_i$ pairs once we assign their label, for $i \leq a^\prime$ by induction.

    Moving onto the second phase, we find elements in $H$ that anti-commute with those in $A$.
    By the same logic as before, the $x_i, z_i$ pairs all satisfy their requirements for $i \leq a^\prime + k$.
    Note that we don't need to check if $[t_k, z_\ell] = 1$ since everything in $A$ commutes with everything in $A$.
    At the end we can set $z_{a^\prime + k + 1}, \dots, z_{a^\prime + b^\prime}$ without worry, as everything in $H$ must commute with everything left in $A$.

    Finally, we only work with remaining elements of $H$ in the last phase and the correctness holds by the same logic as the first phase.
    
    The total runtime is $O(n(a+b)^2)$ as we need to double-iterate through $G$, $A$, and $H$ respectively, leading to $O((a+b)^2) $ many symplectic product calculations.
    Since each symplectic product takes $O(n)$ time to compute we get a total time of $O(n(a+b)^2)$.
\end{proof}

\end{document}